\documentclass[draftcls,onecolumn]{IEEEtran}

\usepackage{amsmath}
\usepackage{amssymb}
\usepackage{graphicx}
\usepackage{algorithm2e}

\usepackage[T1]{fontenc}
\usepackage{url}
\usepackage{ifthen}
\allowdisplaybreaks
\usepackage{amsthm}
\usepackage{bbm}
\usepackage{bm}
\usepackage{hyperref}
\usepackage{float}

\usepackage{placeins}
\usepackage{float}
\usepackage{pgfplots}
\usepackage{pgfplotstable}	
\usepackage{color}

\usepackage{amsfonts}
\usepackage{cite}
\usepackage{dsfont}

\usepackage{adjustbox}

\usepackage{booktabs}

\theoremstyle{theorem}
\newtheorem{theorem}{Theorem}
\newtheorem{lemma}[theorem]{Lemma}
\newtheorem{corollary}[theorem]{Corollary}
\newtheorem{proposition}[theorem]{Proposition}
\theoremstyle{definition}
\newtheorem{definition}{Definition}

\theoremstyle{remark}

\usepackage{threeparttable}
\usepackage{makecell}
\usepackage{multirow}

\usepackage{colortbl}
\definecolor{bgcolor}{rgb}{0.93,0.99,1}
\definecolor{bgcolor2}{rgb}{0.8,1,0.8}
\definecolor{bgcolor3}{rgb}{0.50,0.90,0.50}
\usepackage{tcolorbox}
\usepackage{pifont}
\definecolor{mydarkgreen}{RGB}{39,130,67}
\definecolor{mydarkred}{RGB}{192,25,25}

\usepackage{color} 
\usepackage{blkarray, bigstrut}

\newcommand{\smint}{\mathchoice%
  {\scalebox{0.8}{$\displaystyle\int$}}%
  {\scalebox{0.9}{$\textstyle\int$}}%
  {\scalebox{0.85}{$\scriptstyle\int$}}%
  {\scalebox{0.8}{$\scriptscriptstyle\int$}}%
}

\title
{
Tighter Information-Theoretic Generalization Bounds 
via a Novel Class of Change of Measure Inequalities
}

\usepackage{times}

\author{\IEEEauthorblockN{Yanxiao~Liu*,~Yijun~Fan*~and~Deniz~G\"und\"uz} 
\thanks{* The first two authors contributed equally.}
\thanks{Y.~Liu is with the Department of Electrical and Electronic Engineering, Imperial College London, London, UK. Email: yliu25@ic.ac.uk}
\thanks{Y.~Fan is with the Department of Information Engineering, The Chinese University of Hong Kong, Hong Kong, China. Email: yijunfan@ie.cuhk.edu.hk}
\thanks{D.~G\"und\"uz is with the Department of Electrical and Electronic Engineering, Imperial College London, London, UK. Email: d.gunduz@imperial.ac.uk}
}

\begin{document}

\maketitle
    
\begin{abstract}

Change of measure inequalities translate divergences between probability measures into explicit bounds on event probabilities, and play an important role in deriving probabilistic guarantees in learning theory, information theory, and statistics. 
We propose novel change of measure inequalities via a unified framework based on the data processing inequality, which is surprisingly elementary yet powerful enough to yield novel, tighter inequalities. 
We provide change of measure inequalities in terms of a broad family of information measures, including $f$-divergences (with Kullback-Leibler divergence and $\chi^2$-divergence as special cases), R\'enyi divergence, and $\alpha$-mutual information (with maximal leakage as a special case). 
We apply these results to generalization error analysis, PAC-Bayesian theory, differential privacy, and data memorization, obtaining stronger guarantees while recovering best-known results through simplified analyses. 
\end{abstract}

\section{Introduction}
\label{sec::intro}

In information theory and learning theory, a recurring task is to bound the probability $P(E)$ of an event $E$ in terms of $Q(E)$ and some divergence between $P$ and $Q$.
That is, for a measurable set $E\in\mathcal{F}$ and probability measures $P,Q$ on a common measurable space $(\mathcal{X}, \mathcal{F})$ such that $P\ll Q$\footnote{$P\ll Q$ denotes that $P$ is absolutely continuous with respect to $Q$, and ${\mathrm{d}P}/{\mathrm{d}Q}$ denotes the Radon-Nikodym derivative.}, we aim to establish an inequality of the form 
\begin{equation}
    P(E) \leq \xi \big(Q(E), \mathrm{d}P/\mathrm{d}Q\big),
    \label{eq::change of measure}
\end{equation}
which are referred to as ``change of measure inequalities'' and serve as a bridge between divergences and probabilistic guarantees.
Their applications span a broad range of problems, including generalization error bounds of learning algorithms~\cite{xu2017information, steinke2020reasoning, chu2023unified}, hypothesis testing~\cite{polyanskiy2025information}, distributed detection~\cite{chen2019anonymous}, and Monte Carlo estimation~\cite{ohnishi2021novel}.
A classical example is the strong converse lemma~\cite{polyanskiy2025information}: 
\begin{equation}
    P(E)\leq \gamma Q(E) + P\left( \mathrm{d}P / \mathrm{d}Q >\gamma\right), \quad\forall \gamma \in\mathbb{R}, \label{eq::strong_converse}
\end{equation}
which serves as a fundamental tool in information theory~\cite{cover1999elements} and hypothesis testing~\cite{polyanskiy2025information}.

Recent research has increasingly focused on change of measure inequalities and their applications, with several key studies establishing frameworks for deriving such inequalities.
Notably, \cite{hellstrom2020generalization} applied the strong converse lemma~\eqref{eq::strong_converse} to derive a wide range of generalization error bounds.
This line of work was further extended by~\cite{ohnishi2021novel, picard2022change}, which utilized variational representations of $f$-divergences to obtain tighter change of measure inequalities. 
In~\cite{esposito2021generalization}, change of measure inequalities were used to derive novel results on generalization error and privacy analyses.

While these contributions significantly tightened the resultant bounds, they remain largely isolated without a clear connective bridge, and the strong converse lemma used in~\cite{hellstrom2020generalization} was not involved in the other preceding literature. 
The results in~\cite{esposito2021generalization} are convenient to use, but not as tight as those in~\cite{picard2022change, ohnishi2021novel}, which, in some cases, rely on additional structural assumptions such as strict convexity. Though representing the sharpest known bounds, the results in~\cite{picard2022change} involve auxiliary parameters that are nontrivial to optimize.
Most importantly, the literature currently lacks a unified framework that integrates existing methodologies to provide a comprehensive analysis.

In this work, we propose a unified framework for deriving change of measure inequalities via a single elementary tool: the \emph{data processing inequality (DPI)} for $f$-divergences~\cite{sason2016f}.
Although the DPI for $f$-divergences has been used in previous research, its optimality in bounding event probabilities via change of measure inequalities remains elusive.
In this paper, we unify existing approaches via the application of DPI through an indicator channel, which is surprisingly elementary yet powerful enough to yield novel inequalities that are usually tighter than existing ones. We further establish the optimality of this DPI approach for deriving change of measure inequalities in the form of~\eqref{eq::change of measure}. 
Besides being convenient to use, one advantage of our approach is its flexibility in choosing the error event $E$ in~\eqref{eq::change of measure}, which enables our framework to be applied in diverse problems, e.g., generalization analysis~\cite{xu2017information} and data memorization~\cite{attias2024information, feldman2025trade}. 
A more detailed literature review is in Appendix~\ref{app::literature_review}.

\section{Our Contribution}
\label{sec::contri_related_wk}

We first present the key technique of our paper, demonstrate its use by deriving an tighter change of measure inequality compared with~\eqref{eq::strong_converse}, and then summarize our contributions.

According to DPI, given any ``channel'' $T$ from $\mathcal{X}$ to $\mathcal{Y}$, applying the same $T$ to probability measures $P,Q$ on $\mathcal{X}$ cannot increase their $f$-divergence, that is,
\[
D_f(T\circ P \Vert T\circ Q) \leq D_f(P\Vert Q).
\]
Now, if we apply DPI for the indicator channel $T=\mathds{1}_E$, and let $p:= P(E), q:= Q(E)$, we obtain
\begin{equation}
    D_f(P\Vert Q) \geq D_f(\mathds{1}_E\circ P \Vert \mathds{1}_E\circ Q)  = D_f(\mathrm{Ber}(p)\Vert \mathrm{Ber}(q)) = qf \Big(\frac{p}{q}\Big) + \bigl(1-q\bigr) f \Big(\frac{1-p}{1-q}\Big).\label{eq::DPI_f_div_detail}
\end{equation}
Once $f$ is specified, the upper bound on $p$ relies on expressing the right hand side of~\eqref{eq::DPI_f_div_detail} in terms of $p$. This single inequality is the main tool behind all of our results.

Though surprisingly simple, to highlight the utility of our DPI framework, consider the $E_\gamma$-divergence~\eqref{eq::def_Egamma} (a.k.a. the hockey-stick divergence~\cite{sason2016f}) that is a generalization of the total variation distance and is also an $f$-divergence.
We obtain the following change of measure inequality. 

\begin{proposition}
    \label{prop::E_gamma_converse}
    Fix probability measures $P, Q$ on $\mathcal{X}$ such that $P\ll Q$. 
    For all measurable $E$, 
    \begin{equation}
        P(E)\leq \gamma Q(E) + E_\gamma (P\Vert Q), \quad\forall \gamma \in\mathbb{R}.
        \label{eq::Egamma_change_measure}
    \end{equation}
\end{proposition}
{\renewcommand{\proofname}{Proof Sketch}%
\begin{proof}
    $E_\gamma$ is an $f$-divergence with $f(t)=[t-\gamma]_+$. 
    Let $p:=P(E), q:=Q(E)$. By~\eqref{eq::DPI_f_div_detail},
     \[
    E_\gamma(P\Vert Q) \geq E_\gamma(\mathds{1}_E\circ P \Vert \mathds{1}_E\circ Q) = q\cdot \left[p/q - \gamma\right]_+ + (1-q)\cdot \left[(1-p)/(1-q)-\gamma\right]_+,\,\,\, \forall \gamma \in\mathbb{R}.
    \]
    Since $(1-q)[\cdot]_+\geq 0$, it follows that $E_\gamma(P\Vert Q) \geq q(p/q -\gamma)$, and rearrangement yields~\eqref{eq::Egamma_change_measure}. 
\end{proof}}
We observe that Proposition~\ref{prop::E_gamma_converse} strictly improves upon the strong converse lemma\footnote{See Appendix~\ref{app::E_gamma_converse} for proofs and discussions.
In short, our improvement comes from $E_\gamma (P\Vert Q)\leq \int_{\{\mathrm{d}P / \mathrm{d}Q>\gamma\}} \mathrm{d}P$, i.e., \eqref{eq::strong_converse} only counts \emph{how often} $\mathrm{d}P / \mathrm{d}Q$ exceeds $\gamma$, while Proposition~\ref{prop::E_gamma_converse} \emph{also} quantifies \emph{how far above} $\gamma$ it typically is. 
}~\eqref{eq::strong_converse}.
Considering the special role of the $E_\gamma$-divergence in differential privacy~\cite{asoodeh2021local} and quantum information theory~\cite{hirche2023quantum, hirche2024quantum}, Proposition~\ref{prop::E_gamma_converse} can be used immediately to tighten existing bounds in these problems. 

We then summarize our contributions from various aspects as follows. 

\paragraph{Change of Measure Inequalities.}
We employ the machinery of~\eqref{eq::DPI_f_div_detail} together with other techniques to derive change of measure inequalities in terms of a large family of information measures. Our methods significantly simplify the analysis without sacrificing the tightness. In contrast, our results recover, and usually improve, existing inequalities in terms of the same information measure (see Table~\ref{tab:com_event_specialized_with_dpi_partial}), which seems too good to be true at first glance. However, we can prove the \emph{optimality} of our framework under reasonable assumptions, and we show how to recover existing methods~\cite{picard2022change, esposito2021generalization}, not just their results, through the simplified routes in Section~\ref{subsec::discuss_DPI}.

\paragraph{Applications}
The generality of our result yields strong bounds across a diverse spectrum of fundamental problems in learning theory and other areas.
Our framework provides a unified machinery for these problems by specifying $P$, $Q$, and the event $E$ according to the scenario.
For example, when analyzing the generalization error of a learning algorithm $P_{W|S}$, we specify $P$ as the joint distribution $P_{SW}$, $Q$ as the product distribution $P_SP_W$, and $E$ as the event that the error exceeds a threshold, which gives a path to translate change of measure inequalities into generalization bounds.

\begin{itemize}
    \item \textbf{Generalization Analysis.} One central challenge in machine learning is to quantify the ``generalization'' guarantees of learning algorithms: if an algorithm performs well on the training data, will it also perform well on new samples?
Information-theoretic generalization bounds have been widely studied over the past decade~\cite{russo2016controlling, xu2017information, asadi2018chaining, steinke2020reasoning, sefidgaran2022rate, chu2023unified}. 
In Section~\ref{sec::gen_bd}, we translate our change of measure inequalities into generalization bounds in terms of a wide range of information measures, which either recover existing ones~\cite{esposito2021generalization, ohnishi2021novel, picard2022change} or are novel, and we exemplify our improvement in certain regimes under learning settings.

\item \textbf{Differential Privacy.} Generalization is guaranteed when ``an algorithm leaks little information about its dataset''~\cite{bassily2018learners}, which highlights the connection between generalization and privacy.
Generalization bounds for pure differential privacy (DP)~\cite{dwork2006calibrating} algorithms have been studied~\cite{rodriguez2021upper, liu2026generalization}, but the approximate DP case appeared to be much harder~\cite{rogers2016max}.
By establishing a novel connection between $E_\gamma$-divergence and approximate max-information~\cite{rogers2016max}, we derive a new generalization bound for approximate DP algorithms.

\item \textbf{Data Memorization.} Privacy asks whether an algorithm's output reveals sensitive training data, usually at an individual-level, while \emph{memorization} asks, in a more operational sense, whether the model has retained specific training examples and can reveal or exploit them later.
We apply our change of measure inequalities to a memorization setting~\cite{sefidgaran2025tighter} by specifying another $E$, which can strengthen and generalize existing memorization results. 
\end{itemize}

\paragraph{Notation} \vspace{-5pt}
We use calligraphic letters (e.g., $\mathcal{W}$), capital letters (e.g., $W$), and lower-case letters (e.g., $w$) to denote sets, random variables, and instances, respectively. 
We assume logarithms are base $e$, $[x]_+ \hspace{-1pt}:=\hspace{-2pt} \max\{x, 0\}$, $\mathds{1}_{\{S\}}$ is an indicator of a statement $S$, and $1/\infty\hspace{-1pt} = \hspace{-1pt}0$.  
When distributions $P,Q$ are discussed together, we assume they are on a measurable space $(\mathcal{X}, \mathcal{F})$ where $\mathcal{F}$ is suppressed.

\section{Information Measures}
\label{sec::info_measure}
 \vspace{-4pt}

Information measures are used to quantify the dissimilarity between distributions or the correlation between random variables. 
One of the most general notion is the $f$-divergence: 
\begin{definition}
    Let $f:[0,\infty) \rightarrow\mathbb{R}$ be convex and $f(0) = \lim\limits_{t\, \downarrow\, 0} f(t)$. 
    For $P\ll Q$, define \vspace{-4pt}
    \begin{equation*}
        D_f(P\Vert Q) : = \smint f\left({\mathrm{d}P}/{\mathrm{d}Q}\right)\mathrm{d}Q. 
    \end{equation*}
    \end{definition}

$f$-divergences~\cite{csiszar1967information, csiszar1963informationstheoretische, ali1966general, sason2016f}
include most well-known information measures as special cases, e.g., Kullback-Leibler divergence and total variation distance, and possess desirable properties such as non-negativity, joint convexity, and, more importantly, the data processing inequality. 

$E_\gamma$-divergence is of special interests due to Proposition~\ref{prop::E_gamma_converse} and its use across our paper. 
It is defined: \vspace{-4pt}
\begin{equation}
	E_\gamma(P\Vert Q) := \sup_{E} \big( P(E) - \gamma Q(E) \big) = \smint \left( \left[ {\mathrm{d}P}/{\mathrm{d}Q} - \gamma \right]_+
    \right) \mathrm{d}Q. \label{eq::def_Egamma}
\end{equation}
\begin{definition}
    Let $\alpha\in(0,\infty)\setminus\{1\}$ and $P\ll Q$, \emph{R\'enyi divergence of order $\alpha$} is given by \vspace{-4pt}
    \begin{equation}
        D_\alpha(P\Vert Q)
        :=
        {1}/{(\alpha-1)}\cdot \log \smint \left({\mathrm{d}P}/{\mathrm{d}Q}\right)^\alpha \mathrm{d}Q.
        \label{eq::def_Renyi}
    \end{equation}
\end{definition}
Note that power divergence $\mathcal{H}_\beta (P\Vert Q)$ has a one-to-one transformation with $D_\alpha(P\Vert Q)$, by\vspace{-2pt}
\begin{equation}
    D_\alpha(P\Vert Q) =\log \big(1+(\alpha-1)\mathcal{H}_\alpha(P\Vert Q)\big) \big/(\alpha - 1) . \label{eq::Hellinger_to_Renyi}
\end{equation}

Another important information measure is the Sibson \emph{$\alpha$-mutual information} $I_\alpha(S, W)$, initiated from the ``information radius''~\cite{sibson1969information} and revisited by~\cite{verdu2015alpha}. It is defined as follows.
\begin{definition}
    Let $S,W \hspace{-1pt}\sim P_{SW}$ and $Q_W$ be a probability measure on $\mathcal{W}$, \emph{$\alpha$-mutual information} is \vspace{-3pt}
    \begin{equation}
        I_\alpha(S, W) := \min_{Q_W} D_\alpha\big(
        P_{SW}\Vert P_SQ_W
        \big), \quad \alpha > 0. 
        \label{eq::def_alpha_MI}
    \end{equation}
\end{definition}
Taking $\alpha \rightarrow 1$ recovers $I(S;W)$ and $\alpha \rightarrow \infty$ recovers \emph{maximal leakage} (that will be used in~\eqref{eq::highprob_ML_genbd}): 
\begin{equation}
    \mathcal{L}(S\rightarrow W) = \sup_{U-S-W-\hat{U}} \log \Big({\mathbf{P}(U=\hat{U})}\big/{\max_{u\in\mathcal{U}} \mathbf{P}_U(u)}\Big), \label{eq::def_ML}
\end{equation}
where the supremum is over all $U$ and $\hat{U}$ taking values in the
same finite alphabet~\cite{issa2016operational}.

\section{Change of Measure Inequalities}
\label{sec::change_mea_ineq}

We have briefly explained our main technique in Section~\ref{sec::contri_related_wk}, namely the \emph{data processing inequality (DPI)} for $f$-divergences~\cite[Section~7.2]{polyanskiy2025information} that is surprisingly elementary yet remarkably powerful for yielding strong change of measure inequalities. 
Here we present it again for completeness. 
\[
D_f(T\circ P \Vert T\circ Q) \leq D_f(P\Vert Q). 
\]

To derive change of measure inequalities of form~\eqref{eq::change of measure}, we consider the deterministic indicator channel $T =\mathds{1}_E$ (see Section~\ref{subsec::discuss_DPI} for the intuition behind). 
Denoting $p:= P(E)$ and $q:= Q(E)$, we observe
\begin{equation*}
    D_f(P\Vert Q) \geq D_f(\mathds{1}_E\circ P \Vert \mathds{1}_E\circ Q)  = D_f(\mathrm{Ber}(p)\Vert \mathrm{Ber}(q)) = qf \Big(\frac{p}{q}\Big) + \bigl(1-q\bigr) f \Big(\frac{1-p}{1-q}\Big),
\end{equation*}
and the machinery of proof we use is simply to specify $f$ and calculate the last term of~\eqref{eq::DPI_f_div_detail}.

We have used Proposition~\ref{prop::E_gamma_converse} to demonstrate how to use our machinery to derive change of measure inequalities of form~\eqref{eq::change of measure}. 
By the same pipeline, we derive change of measure inequalities in terms of various information measures in Theorem~\ref{thm::hellinger_chi2}; the proof is in Appendix~\ref{sec::proof_prop_hellinger_chi2}. 
\begin{theorem}
    \label{thm::hellinger_chi2}
    For probability measures $P, Q$ on $\mathcal{X}$ such that $P\ll Q$, for all measurable $E$, 
    \begin{align}
        & P(E)
        \leq \big(D_{\mathrm{KL}}(P\Vert Q)+\log (1+Q(E)(e^{c}-1))\big)/c,\quad c>0, 
        \label{eq::PQ_KL} \\
        & P(E) \leq Q(E) + \sqrt{Q(E)(1-Q(E))\chi^2(P\Vert Q)}, \label{eq::chi2}  \\
        & 2\left(1 - \sqrt{P(E)Q(E)} - \sqrt{(1-P(E))(1-Q(E))}\right) \leq H^2(P;Q), \label{eq::PQ_Hellinger_distance}  \\
        & 
        P(E)^\beta Q(E)^{1-\beta} + (1 - P(E))^\beta(1-Q(E))^{1-\beta}
        \leq
        1+(\beta-1)\mathcal H_\beta(P\Vert Q). \label{eq::PQ_Hellinger_divergence2} 
    \end{align}
\end{theorem}
Note that one can derive cleaner forms by further relaxations, see Appendix~\ref{sec::proof_prop_hellinger_chi2}. 
By~\eqref{eq::Hellinger_to_Renyi}, we can derive inequalities in terms of the R\'enyi divergence~\eqref{eq::def_Renyi}. 
In Table~\ref{tab:com_event_specialized_with_dpi_partial}, we list the results from our DPI approach, for typical $f$-divergences, and a table to explicitly compare with ~\cite{picard2022change} is in Appendix~\ref{app::table}.


For a complete, $\sigma$-finite probability space $(\Omega,\mathcal{F}, \mu)$, for $U$ that is $\mathcal{F}$-measurable, with respect to $\mu$, the \emph{Luxemburg norm} of $U$ is defined as $\Vert U\Vert_\psi^\mu := \inf\big\{ \sigma>0: \mathbf{E}_\mu \big[\psi \big(|U|/\sigma \big) \big]\leq 1 \big\}$ and the \emph{Amemiya norm} of $U$ is defined as $\Vert U\Vert_\psi^{A, \mu} := \inf\big\{ \big(\mathbf{E}_\mu\big[\psi\big(t|U| \big) \big] + 1\big) / t: t>0 \big\}$. 
Recall for  a convex $\psi$, its \emph{conjugate} $\psi^\star :[0,\infty) \rightarrow\mathbb{R}$ is defined as $\psi^\star(t) = \sup_{\lambda>0} \lambda t - \psi(\lambda)$. 
An \emph{Orlicz function} $\psi$ is a convex function $\psi:[0,\infty) \rightarrow[0,\infty]$ that vanishes at zero and is not identically $0$ or $\infty$ on $(0, \infty)$, and its generalized \emph{inverse} is defined as $\psi^{-1}(s) :=\inf\big\{ t\geq 0: \psi(t)\geq s \big\}$ for $s\geq 0$. 

\begin{theorem}
\label{thm::E_gamma_converse_generalized_Orlicz}
    Fix $P, Q$ such that $P\ll Q$. 
    For all measurable $E$, fix an Orlicz function $\psi$, 
     \begin{equation}
     \label{eq::E_gamma_converse_pq_Orlicz}
         P(E)\leq \gamma Q(E) + 
         \frac{1}{\psi^{-1}(1/Q(E))}
         \cdot \bigg\Vert  \left[\frac{\mathrm{d}P}{\mathrm{d}Q} - \gamma\right]_+ \bigg\Vert_{\psi^\star}^{A, Q}, \quad\forall \gamma \in\mathbb{R}.
     \end{equation}
\end{theorem}
See Appendix~\ref{app::E_gamma_converse_generalized_Orlicz} for a proof and discussions.

\begin{table}[!htbp]
\small
\caption{
Our change of measure inequalities in terms of typical $f$-divergences via DPI, all of which are never worse, and usually tighter, than \textbf{best-known} inequalities in literature~\cite{picard2022change, esposito2021generalization, ohnishi2021novel}. 
}
 \vspace{-0.5em}

\renewcommand{\arraystretch}{1.65}
\centering
\begin{tabular}{p{2.248cm} p{2cm} p{6.9cm} p{1cm}}
\hline\hline 
$f$-div
& $f(t)$
& Change of Measure Inequalities via DPI
& Tighter?
\\
\hline\hline

\parbox[t]{2.35cm}{$E_\gamma$-div, $\gamma\geq 1$}
& $\displaystyle [t-\gamma]_+$
& $\displaystyle p\leq \gamma\,q+E_\gamma(P\Vert Q)$
& $\;$ new
\\

KL
& $t\log t$
& $\displaystyle p\leq \big(D_{\mathrm{KL}}(P\Vert Q)+\log  \big(1+q(e^c-1)\big)\big)\big/ c$, $c>0$
& $\;$ same
\\

$\chi^2$-div
& $t^2-1$
& $\displaystyle p\leq q+\sqrt{q\bigl(1-q\bigr)\chi^2(P\Vert Q)}$
& $\;$ same
\\

\parbox[t]{2.35cm}{Power-$\beta$, $\beta>1$}
& $\displaystyle (t^\beta-1) / (\beta-1)$
& $\displaystyle \begin{aligned}[t]
&p^\beta q^{1-\beta} +(1-p)^\beta(1-q)^{1-\beta} \leq 1+(\beta-1)\mathcal H_\beta (P\Vert Q)
\end{aligned}$
& $\;\;\;\; \checkmark$
\\

Squared Hellinger
& $(1-\sqrt{t})^2$
& $\displaystyle 2\Bigl(1-\sqrt{pq}-\sqrt{(1-p)(1-q)}\Bigr)\leq H^2(P;Q)$
& $\;\;\;\; \checkmark$
\\

Reverse $\chi^2$-div
& $\displaystyle 1/t - 1$
& $\displaystyle {(p-q)^2}/(p(1-p))\leq \chi^2(Q\Vert P)$
& $\;\;\;\; \checkmark$
\\

Reverse-KL
& $-\log t$
& $\displaystyle \begin{aligned}[t]
&q\log(\frac{q}{p}) +(1-q)\log(\frac{1-q}{1-p}) \leq D_{\mathrm{KL}}(Q\Vert P)
\end{aligned}$
& $\;\;\;\; \checkmark$
\\ 

Jensen-Shannon
& $\displaystyle \log \frac{2^{t+1}t^t}{(1+t)^{t+1}}$
& $\displaystyle 2h_2  \left((p+q)/2\right)-h_2  \bigl(p\bigr)-h_2  \bigl(q\bigr)\leq \mathrm{JS}(P\Vert Q)$
& $\;\;\;\; \checkmark$
\\

Vincze-Le Cam
& $\displaystyle (2-2t)/(t+1)$
& $\displaystyle 2(p-q)^2/((p+q)\bigl(2-p-q\bigr))\leq \mathrm{VC}(P;Q)$
& $\;\;\;\; \checkmark$
\\
\hline
\end{tabular}
\label{tab:com_event_specialized_with_dpi_partial} 
\vspace{-1em}
\end{table}

\subsection{Discussions}
\label{subsec::discuss_DPI}

We then discuss the intuition of our DPI framework, and show its relation to the existing approaches \cite{picard2022change, esposito2021generalization}, which are the tightest known results in this context to the best of the authors' knowledge.

The idea behind our DPI scheme is as follows: the discrepancy between $P(E)$ and $Q(E)$ depends on two ingredients: (i) the choice of event $E$; (ii) the statistical dissimilarity between $P$ and $Q$. 
For example, if $E=\Omega$, then the event selection eliminates any observable difference. 
We interpret $P(E)$ and $Q(E)$ as the output probabilities induced by passing $P$ and $Q$ through the same kernel $T=\mathds{1}_E$, and the difference between $P(E)$ and $Q(E)$ can then be quantified by information measures between the induced laws $P\circ T$ and $Q\circ T$, implying~\eqref{eq::DPI_f_div_detail}. 
Though the DPI was also used in a step of~\cite[Theorem~3]{esposito2021generalization}, given the explicit expression of $f$, one can provide a tighter bound on $f((1-p)/(1-q))$ than $-f^*(0)$ (see \eqref{eq::compare_thm3}), which is one of the motivations for deriving the results in Theorem~\ref{thm::hellinger_chi2}.

We note that DPI has been applied in a few studies on generalization error as a convenient tool rather than an integrative framework. 
However, when applied to bound $P(E)$ by $Q(E)$ and $D_f(P\Vert Q)$ of some measurable $E$, DPI is tighter than the relaxation of the Donsker-Varadhan representation
\[
D_f(P\Vert Q) = \sup_{T\in\mathcal T} \left\{\mathbf E_P[T]-\mathbf E_Q[f^*(T)]\right\},
\]
and it is the best way to perform:  
To bound $P(E)$, $T$ is restricted to $\mathcal T_E:=\big\{\,T:\ T=a\mathds 1_E+b\mathds 1_{E^c}\text{ for }a,b\in\mathbb R\,\big\}$. In this case, for every measurable $E$, 
\[
\sup_{T\in\mathcal T_E}
\left\{\mathbf E_P[T]-\mathbf E_Q[f^*(T)]\right\} = D_f\big(\mathrm{Ber}(P(E))\Vert \mathrm{Ber}(Q(E))\big),
\]
i.e., to compare $P, Q$ via $E$, optimizing over the admissible test functions gives \emph{exactly} Bernoulli divergence. 
Hence, DPI with $x\mapsto \mathds 1_E(x)$ gives the \emph{best possible} result. 
The inequality here comes from the loss of information from the indicator mapping $\mathds{1}(E)$ in measuring $P(E)$, instead of the relaxation of $\mathbf E_Q[f^*(T)]$, see Appendix~\ref{app::DPI_optimal} for details. We recover existing approaches as follows:



\begin{itemize}
	\item \textbf{Recovering~\cite{picard2022change}}: the DPI \eqref{eq::DPI_f_div_detail} gives 
	\[
	G_f(p,q):= qf({p}/{q}) + \bigl(1-q\bigr) f ({(1-p)}/{(1-q)}) \leq D_f(P\Vert Q),
	\]
	and applying Fenchel duality, $f(x)=\sup_{u\in\mathbb R}\{ux-f^*(u)\}$, to the two terms of $G_f(p,q)$: 
	\[
	G_f(p,q)
	=
	\sup_{u,v}\big\{p(u-v)+v-qf^*(u)-(1-q)f^*(v)\big\}.
	\]
	Let $u=\lambda+c$ and $v=c$ for $\lambda>0$ and $c\in\mathbb{R}$, we then recover \cite{picard2022change} on $\mathds{1}_E$:
	\[
	p\le
	\inf_{u>v}
	{\big(D_f(P\Vert Q)-v+qf^*(u)+(1-q)f^*(v)\big)}{\big/ (u-v)}.
	\] 
	Comparing to \cite{picard2022change}, our DPI scheme is simpler and $G_f(p,q)$ is easier to control, which gives us a better control over the inequalities and eliminates the parameters in \cite[Table~1]{picard2022change} that can be hard to optimized, hence gives tighter inequalities. 
	Moreover, their $f$ has to be strictly convex, hence they cannot recover our Proposition~\ref{prop::E_gamma_converse} in terms of $E_\gamma$ divergence. 
	
	\item \textbf{Recovering~\cite{esposito2021generalization}}: We rearrange~\eqref{eq::DPI_f_div_detail} and recover \cite[Theorem~3]{esposito2021generalization} by  
	\begin{align}
	f \big ({p}/{q}\big)  
	& \leq {\big (D_f(P\Vert Q)-(1-q)f\big((1-p)/(1-q)\big)\big)}\big/ q \nonumber  \\
	& \leq \big(D_f(P\Vert Q)+(1-q)f^*(0)\big) \big/ q, \label{eq::compare_thm3}
	\end{align}
	where the last inequality is by $f^*(0) = -\inf_{t\geq 0} f(t) \geq -f\big((1-p)/(1-q)\big)$. 

\end{itemize}

We note that the concurrent work~\cite{guan2025DPI} used a similar approach with the DPI of $f$-divergence. 
However, they considered PAC-Bayesian bounds only, and in this case our bounds are usually tighter than theirs (e.g., compare~\eqref{eq::PQ_Hellinger_divergence2_relaxed2} to~\cite[Lemma~2]{guan2025DPI}). Moreover, our results can be applied to $\sigma$-sub-Gaussian loss, instead of their bounded $[0,1]$-loss. 
Most importantly, the implication of this study is much broader: based on DPI, we establish a unified framework that is applicable to generalization, privacy and memorization problems, and one key result (Proposition~\ref{prop::E_gamma_converse}) is novel and may lead to further results in different areas.

\section{Generalization Error Bounds}
\label{sec::gen_bd}

\subsection{Generalization Error Bounds via $f$-Divergence}
\label{sec::f_div}

We first briefly review the background, and then apply our change of measure inequalities.

Consider a stochastic learning algorithm $P_{W|S}$ as a probabilistic mapping from a training dataset $S := \left(Z_1,\ldots,Z_n \right)\in\mathcal{Z}^n$, where $Z_1,\ldots,Z_n \sim P_Z $ i.i.d., to hypothesis $W \in \mathcal{W}$. 
With loss function $\ell:\mathcal{W}\times\mathcal{Z}\rightarrow\mathbb{R}$,
the generalization performance of $P_{W|S}$ is measured by its \emph{generalization error}
\[
\mathrm{gen}(S, W) := \mathbf{E}_{P_Z}\left[\ell(Z, W) \right] - \frac{1}{n}\sum^n_{i=1} \ell(Z_i, W),
\]
i.e., the gap between the expected loss $\mathbf{E}_{P_Z}\left[\ell(Z, W) \right]$ and the empirical loss $1/n \cdot \sum^n_{i=1} \ell(Z_i, W)$.

Started from~\cite{russo2016controlling, xu2017information} which bounded the expected generalization error by the mutual information $I(S; W)$ between training data and hypothesis, various types of bounds were subsequently proposed~\cite{asadi2018chaining, steinke2020reasoning, sefidgaran2022rate, chu2023unified, aminian2024learning}. 
We consider high-probability generalization bounds~\cite{esposito2021generalization, hellstrom2020generalization}, e.g., 
\begin{equation}
    \mathbf{P}\left(\big|\mathrm{gen}(S, W)\big|\geq\eta \right) \leq 2\cdot \exp\left(
    \mathcal{L}(S\rightarrow W) - n\eta^2/(2\sigma^2)
     \right),
     \label{eq::highprob_ML_genbd}
\end{equation}
where $\mathcal{L}(S\rightarrow W)$ is maximal leakage~\cite{issa2016operational} and the loss function is assumed to be $\sigma$-sub-Gaussian\footnote{If $X\sim P_X$ is $\sigma$-sub-Gaussian, Hoeffding's inequality gives $\forall \eta > 0$, $P_X\left(|X - \mathbf{E}[X]|\geq \eta \right)\leq 2 \exp ( - {\eta^2}/{2 \sigma^2} )$.}. 

Our change of measure inequalities in Table~\ref{tab:com_event_specialized_with_dpi_partial} can be converted to this class of guarantees by taking
\begin{equation}
    P := P_{SW}, \quad Q := P_{W} P_{S}, \quad E = \left\{ W, S: \left| \mathrm{gen}(S, W) \right|\geq \epsilon \right\}. \label{eq::gen_bd_setting_PQE}
\end{equation}
Take some common $f$-divergences as examples, we derive the following generalization bounds.

\begin{theorem}
    \label{thm::gen_bd_f_div}
    Fix $\gamma\in\mathbb{R}$. 
    Let $P\ll Q$, if convex $f:(0,\infty)\hspace{-2pt} \rightarrow \hspace{-2pt} \mathbb{R}$ satisfies $f(1) = 0$, $f'(\gamma) > f'(1)$ for $\gamma > 1$, then for all $\eta > 0$, let $P = P_{SW}$ and $Q=P_SP_W$, we have $\mathbf{P}\left(
    \left| \mathrm{gen}(S,W) \right| \geq \eta
    \right) \leq$ 
    \begin{equation}
    \min
    \left\{
    \begin{array}{l} 
        \gamma \zeta + E_\gamma(P\Vert Q)
        \\
        \sqrt{\zeta(1-\zeta)\cdot \chi^2(P\Vert Q)} + \zeta
        \\
        \sqrt{1-\zeta} \cdot \sqrt{1- \big(1-{H^2(P;Q)}/{2}\big)^2} \big)^2 + \big(\big(1-{H^2(P;Q)}/{2}\big)\sqrt{\zeta}
        \\
        \big(
        \zeta^{\beta-1}
        \big[
        1+(\beta-1)\mathcal H_\beta(P\Vert Q)
        -(1-\zeta)^{1-\beta}(1-u_0)^\beta
        \big]_+
        \big )^{1/\beta}
    \end{array}
    \right\}. \label{eq::gen_bd_all}
    \end{equation}
where $u_0 := \min\big\{1,\ \big((1+(\beta-1)\mathcal H_\beta(P\Vert Q))\,q^{\beta-1}\big)^{1/\beta}\big\}$ and $\zeta := 2\exp\left(- n\eta^2/ (2\sigma^2) \right)$.\footnote{We upper bound $Q(E)$ by using the Hoeffding's inequality on the $\sigma$-sub-Gaussian loss function and hence have $\zeta$.} 
\end{theorem}

The generalization bounds in Theorem~\ref{thm::gen_bd_f_div} are novel and are usually tighter than known bounds in terms of the same measure.
One purpose of deriving generalization bounds in terms of different measures is that, in different applications, different information measures have desirable properties in different senses.
For example, $\chi^2(P\Vert Q)$ is sometimes significantly smaller than $\mathcal{L}(S\to W)$ since $\chi^2 \left( P_{SW}\Vert P_S P_W\right) \leq \exp \big(\mathcal{L}(S\rightarrow W)\big)-1$~\cite{issa2018computable}, and we use examples (see Appendix~\ref{app::gen_DP}) to show that the bound in terms of $\chi^2(P\Vert Q)$ can be strictly tighter than the best-known generalization bound in terms of maximal leakage~\cite{esposito2021generalization}, i.e., \eqref{eq::highprob_ML_genbd}, for both Gaussian and Laplace mechanisms.

Note that although we focus on high-probability bounds, we can also derive average generalization bounds via change of measure inequalities if they are of interest.
For example, by Proposition~\ref{prop::E_gamma_converse}, we can recover the celebrated average bound of~\cite[Theorem~1]{xu2017information}, up to a multiplicative constant:
\[
\mathbf{E}\left[
        |\mathrm{gen}(S, W)|
        \right] \leq 
        \big(2\sigma/n \big)
        \big(2\sqrt{I(S;W)+2/e} + \sqrt{\pi}\big),
\]
which is tighter than the one recovered by~\cite[Corollary~1]{chu2023unified}, see details in Appendix~\ref{app:proof_gen_bd_MI}.

\subsection{Generalization Error Bounds via Maximal Leakage and $\alpha$-Mutual Information}
\label{sec::ML}

As shown in~\eqref{eq::highprob_ML_genbd}, maximal leakage~\cite{issa2016operational}, as a special case of $\alpha$-mutual information~\cite{verdu2015alpha} of order $\infty$, can be used to bound the generalization error~\cite{esposito2021generalization}.
This provides various desirable properties, including an \emph{exponentially} decaying probability of large generalization error and independence from the sample distribution, which is useful in the analysis of additive-noise settings.



Directly employing our Proposition~\ref{prop::E_gamma_converse} recovers the maximal leakage bound~\eqref{eq::highprob_ML_genbd} \emph{exactly}, with simplified analysis, while \cite{hellstrom2020generalization} can only recover it up to a logarithmic term.  
\begin{corollary}
\label{cor:gamma_ML_event}
For $P_{SW} \hspace{-1pt} \ll \hspace{-1pt} P_S P_W$, $E$ measurable, $w\in\mathcal W$ and
$M(w) = \operatorname*{ess\,sup}_{s\sim P_S}\frac{\mathrm{d}  P_{S|W=w}}{\mathrm{d}  P_S}(s)$, 
\begin{equation}
\label{eq:ML_event_cor3}
P_{SW}(E) \leq \mathbf E_{P_W} \big[M(W) P_S(E_W)\big] \leq  \Big(\operatorname*{ess\,sup}_{w\sim P_W}P_S(E_w)\Big)\,
\exp \big(\mathcal L(S\to W)\big).
\end{equation}
\end{corollary}
The proof is given in Appendix~\ref{app::gen_bd_ML_alpha_MI}.
Interestingly, similar to $f$-divergences, $\alpha$-mutual information also satisfies the data processing inequality.
We can use a similar DPI approach to recover the $\alpha$-mutual information bound~\cite[Corollary~1]{esposito2021generalization} with a simple analysis, and hence also recover~\eqref{eq::highprob_ML_genbd}.
Suppose $P_{Y_\alpha}$ is the minimizer of $I_\alpha(X,Y)$.
Denoting $p:=P_{XY}(E)$ and $q_\alpha := (P_XP_{Y_\alpha})(E)$, the DPI gives 
\begin{equation}
\label{eq:two_point_sibson}
I_\alpha(X, Y) \geq D_\alpha \big(\mathrm{Ber}(p) \Vert \mathrm{Ber}(q_\alpha)\big) = 
1/(\alpha-1)\cdot \log \big( p^\alpha q_\alpha^{1-\alpha} + (1-p)^\alpha(1-q_\alpha)^{1-\alpha} \big),
\end{equation}
and  $\alpha\to \infty$ gives $I_\infty (X;Y) \geq \log (p/q_\infty)$, which recovers Corollary~\ref{cor:gamma_ML_event} (see Appendix~\ref{app::recovery_ML_DPI}).


\subsection{PAC-Bayesian Bounds}
\label{sec::PAC_Bayesian}

Probably approximately correct (PAC)-Bayesian bounds~\cite{mcallester1998some, mcallester1999pac, catoni2004statistical, catoni2007pac} are another important class of generalization bounds, which have witnessed a surge of interest in recent years~\cite{perez2021tighter, guedj2021still, biggs2022non, guedj2019primer, alquier2024user}.
In~\cite{hellstrom2020generalization, ohnishi2021novel, picard2022change, guan2025DPI}, change of measure inequalities have been used to derive PAC-Bayesian bounds, motivating us to apply our novel change of measure inequalities as well. 
In this setting, each time the algorithm is used, a new hypothesis $W$ is drawn from $P_{W|S}$, and we aim to find bounds of the form 
\[
\mathbf{P}_{P_S} \big(\mathbf{E}_{P_{W|S}}\big[\mathrm{gen}(S,W)\big]\leq \epsilon\big) \geq 1-\delta.
\]

We take the change of measure inequality~\eqref{eq::PQ_Hellinger_divergence2}, in terms of the power-$\beta$ divergence $\mathcal{H}_\beta(P\Vert Q)$, as an example to show that  we obtain tighter PAC-Bayesian bound. 
The proof is in Appendix~\ref{app::PAC_baye_Hellinger_div}. 

\begin{corollary}
    \label{cor::PAC_baye_Hellinger_div}
    For a learning algorithm $P_{W|S}$ with a $\sigma$-sub-Gaussian loss function, denoting $\mathsf{H}_\beta := \mathcal{H}_\beta(P_{W|S}\Vert P_W)$, 
    the following holds with probability at least $1-\delta$ under any fixed $P_S$:\vspace{-2pt}
    \begin{align*}
    &\mathbf{E}_{P_{W|S}}[\mathrm{gen}(S,W)] \leq \sqrt{\frac{2\sigma^2}{n}}\left(\log 2
            + \log \sqrt{\frac{\pi\beta}{\beta-1}} + \frac{\beta}{4(\beta-1)} + \log \left(\frac{\left((\beta-1) \mathsf{H}_\beta + 1\right)^{1/\beta}}{\delta}\right)\right).
    \end{align*}
\end{corollary}
The bound above is of order $\mathcal{O}\big(\sigma\left(\frac{1}{n}\right)^{1/2}\log \left(\frac{1}{\delta}\right)\log \big((\mathsf{H}_\beta^{1/\beta} \big)\big)$, improving the result in~\cite{ohnishi2021novel}.

Our setting is similar to the one in~\cite{hellstrom2020generalization}, which provided PAC-Bayesian bounds in terms of the information density; these bounds are not directly comparable to ours.
In~\cite{guan2025DPI}, PAC-Bayesian bounds are also derived using the DPI. Their results, however, only work for loss functions bounded in $[0,1]$.

Moreover, our framework also applies to the conditional mutual information (CMI) framework for generalization~\cite{steinke2020reasoning}, which remains finite and usually provides more stable generalization guarantees.
Our CMI results are novel and are deferred to Appendix~\ref{sec::CMI} due to space limitations. 

\vspace{-2pt}

\section{Generalization and Differential Privacy}
\label{sec::DP}
\vspace{-2pt}

Differential privacy~\cite{dwork2006calibrating, dwork2014algorithmic} is one of the most celebrated privacy measures in the past two decades. 
An algorithm $\mathcal{A} := P_{W|S}$ is said to satisfy $(\varepsilon,\delta)$-differential privacy if, for every pair of \emph{neighboring} datasets $s,s'\in\mathcal{S}$ that differ in exactly one coordinate, and for very measurable set $\mathcal{V}\subseteq\mathcal{W}$, we have\vspace{-2pt}
\[
\Pr \big( \mathcal{A}(s)\in \mathcal{V}\big) \leq e^{\varepsilon}\Pr\big(\mathcal{A}(s')\in \mathcal{V}\big) + \delta.
\]

Differential privacy (DP) is tightly related to generalization analysis since private algorithms leak little information, are relatively ``stable''~\cite{dwork2015generalization, dwork2015preserving}, and generalize well.

For the case of $\delta=0$, which is referred to as pure DP, it has been shown that the generalization error can be readily bounded~\cite{dwork2015generalization, dwork2015preserving, rodriguez2021upper, liu2026generalization}. \cite{bassily2016algorithmic} develops a view of DP from distributional stability and gives a tight characterization of the resulting generalization guarantees.

It becomes more tricky and potentially difficult for the case when $\delta > 0$, which is referred to as \emph{approximate} differential privacy. 
As a relaxation of max-information defined by~\cite{dwork2015generalization}, \emph{approximate max-information} is proposed by~\cite{rogers2016max} to capture the generalization properties. 
\begin{definition}\label{def::approx_max_info}
    The $\tau$-approximate max-information between $S$ and $W$, is defined as\vspace{-2pt}
    \[
    I_\infty^{\tau}(S;W) = \log \sup_{\mathcal{V}\subseteq (\mathcal{S}\times\mathcal{W}), \mathbf{P}((S,W)\in \mathcal{V}) > \tau} \frac{\mathbf{P}((S,W)\in \mathcal{V}) - \tau }{\mathbf{P}(S\otimes W\in \mathcal{V})},
    \]
where $S\otimes W$ denotes a random variable  obtained by drawing independent copies
of $S$ and $W$ from respective marginal distributions.
\end{definition}

Consider the generalization error bounds in Section~\ref{sec::gen_bd} in terms of various information measures. Interestingly, for an algorithm $P_{W|S}$ that is $(\varepsilon,\delta)$-DP with $\delta>0$, denote $P:=P_{WS}$ and $Q:=P_WP_S$, the \emph{only} information measure that admits an explicit, finite, and non-trivial upper bound in terms of $\varepsilon$ and $\delta$ is $E_\gamma(P\Vert Q)$. 
In comparison, $\mathcal{H}_\beta(P\Vert Q)$ and $\mathcal{L}(S\to W)$ can be unbounded, and $\mathrm{TV}(P,Q)$ and $H^2(P,Q)$ only admit trivial bounds. 
Therefore, it is of interest to connect $E_\gamma(P\Vert Q)$ with $I_\infty^{\tau}(S;W)$. 
We thus derive the following connection, whose proof can be found in Apendix~\ref{app::approxmaxinfo_Egamma_equiv}. 
\begin{lemma}
\label{lem:approxmaxinfo_Egamma_equiv}
Let $P,Q$ be probability measures on $\mathcal{X}$. 
Fix $\tau\in[0,1)$. For any $\gamma\geq 1$, $E_\gamma(P\Vert Q)\leq \tau$ if and only if $I_\infty^\tau(P\Vert Q)\leq \log \gamma$. 
\end{lemma}

Similar to~\cite{rogers2016max}, we assume the dataset $S$ is drawn i.i.d.\ from a product distribution (see a converse result in~\cite{rogers2016max} showing this assumption is somewhat necessary). 
Combining Lemma~\ref{lem:approxmaxinfo_Egamma_equiv} and~\cite[Theorem~III.1]{rogers2016max}, we derive the following result. 
\begin{proposition}
    \label{prop::Egamma_tau}
    Let $\mathcal{A}:\mathcal{Z}^n \to \mathcal{W}$ be $(\epsilon,\delta)$-DP with $\epsilon \in (0, \frac{1}{2}]$ and $\delta\in (0, \epsilon)$, there exist constants $c_1, c_2 > 0$ such that by taking $\tau = e^{-\epsilon^2 n} + c_1 n \sqrt{\delta / \epsilon}$ and $k = c_2(\epsilon^2 n + n \sqrt{\delta / \epsilon})$, we have \vspace{-2pt}
    \begin{equation*}
        E_{e^k}(P\Vert Q)\leq \tau.
    \end{equation*}
\end{proposition}
We hence derive the following generalization bound by combining Proposition~\ref{prop::Egamma_tau} with  Proposition~\ref{prop::E_gamma_converse}: if $P_{W|S}$ is~$(\epsilon,\delta)$-DP and satisfies the conditions in Proposition~\ref{prop::Egamma_tau}, then for any $\eta \in (0,1)$,\vspace{-2pt}
\begin{equation*}
\mathbf{P}
\left( \big| \mathrm{gen}(S, W)\big|\geq\eta \right) \leq 
2 \exp\left(c_2(\epsilon^2 n + n \sqrt{\delta / \epsilon}) - n\eta^2/(2\sigma^2)
 \right) + e^{-\epsilon^2 n} + c_1 n\sqrt{\delta / \epsilon}.
\end{equation*}

\section{Data Memorization}
\label{sec::memorization}

In machine learning, a critical question is \emph{data memorization}, where the model output contains specific training information beyond what is needed to learn the underlying data distribution.
Unlike privacy, which asks whether an algorithm's output reveals sensitive information about the training data and is usually a guarantee on the algorithm itself, memorization focuses on whether the trained model has actually retained specific training examples and can reveal or exploit them later, often as measured by post-training information quantities.

In this section, by exploiting the flexibility of the event choice in our DPI framework, we apply our change of measure inequalities to derive novel results in a data memorization setting~\cite{attias2024information, sefidgaran2025tighter}.

The framework is similar to the CMI setting (see Appendix~\ref{sec::CMI}). 
Let
$\tilde Z=(Z_1,\ldots,Z_{2n})$ be i.i.d. from $P_Z$, and
$J=(J_1,\ldots,J_n)$ have i.i.d. $\mathrm{Bern}(1/2)$ entries, independent of
$\tilde Z$. 
Define: \vspace{-2pt}
\[
\text{ for } j\in\{0,1\}, i\in[n], \qquad Z_i(j):=\tilde Z_{i+jn},
\qquad
Z(J):=(Z_1(J_1),\ldots,Z_n(J_n)).
\]
The algorithm induces $P_{W|Z(J)}$. Let $Y=(W,\tilde Z,U)$, where $U$ is an independent auxiliary randomness, and let
$\hat J=\phi(Y)\in\{0,1\}^n$ be any estimator of $J$. For
$b\in\{0,\ldots,n\}$, define\vspace{-2pt}
\[
p_b:=\mathbf P\big(d_H(\hat J,J)\le b\big),
\qquad
q_b:=2^{-n}\sum_{k=0}^{b}\binom{n}{k},
\]
where $p_n$ is the actual success probability of the adversary, and $q_b$ is the the null success probability.

We use $P_0:=P_{W\tilde ZJ}, Q_0:=P_{W|\tilde Z}P_{\tilde ZJ}$ to measure the information memorized, e.g., CMI~\cite{sefidgaran2025tighter, feldman2025trade} is captured by $D(P_0\Vert Q_0)=I(W;J| \tilde{Z})$. 
In general, DPI applied on event $\{d_H(\hat J,J)\le b\}$ gives
\begin{equation}
D_f(P_0\Vert Q_0)
\ge
D_f\big(\mathrm{Ber}(p_b)\Vert \mathrm{Ber}(q_b)\big). 
\label{eq::selector_binary_DPI}
\end{equation}
If we specialize \eqref{eq::selector_binary_DPI} to the KL divergence (also see Table~\ref{tab:com_event_specialized_with_dpi_partial}), then for every $c>0$,\
\begin{equation}
p_b
\leq 
\big(I(W;J| \tilde{Z})+\log\big(1+q_b(e^c-1)\big)\big)\big/{c}.
\label{eq::selector_KL_bound}
\end{equation}
This already recovers the Fano step used in
\cite[Theorem~5]{sefidgaran2025tighter} by taking $c=\log(1/q_b)$. 
The advantage of \eqref{eq::selector_KL_bound} is that $c$ can be optimized: with a better choice of $c$, let $t<1/2$, we derive 
\[
\mathbf P\big(d_H(\hat J,J))\le g(I(W;J| \tilde{Z})) + O\left((\log n)/n^2\right), 
\]
for some function $g$. 
By contrast, the corresponding remainder term in \cite{sefidgaran2025tighter} is of order $O(1/n)$. 
Thus, in the low-CMI regime, we improve their bound by reducing the slack to $O((\log n)/n^2)$.

We then have the memorization result that can be stronger than~\cite[Theorem~5]{sefidgaran2025tighter}. 
\begin{theorem}
\label{thm::mem}
Let $Q:\mathcal W\times\mathcal Z\times\mathcal P(\mathcal Z)\to\{0,1\}$ be a membership query, i.e., $Q(W,z,\mu)=1$ if adversary says $z$ is a training sample. 
By relabeling, let $Z_{i,1}$ and $Z_{i,0}$ be training and fresh sample, respectively, take $T_n:=\sum_{i\in[n]}Q(W,Z_{i,1},\mu)$ and $F_n:=\sum_{i\in[n]}Q(W,Z_{i,0},\mu)$, if $\mathbf P(T_n\ge \alpha n)\ge q$ for some $\alpha\in(0,1)$, then for every $\beta\in[0,\alpha)$ and every $\eta\in\left(0,{(\alpha-\beta)}/{2}\right)$, we have
\[
\mathbf P(F_n>\beta n)
\ge
q
-
{
\Psi_n\big(\frac{1-\alpha+\beta}{2}+\eta\big)
}\Big/{
\Big(1-\exp\big(-\frac{2\eta^2}{1-\alpha+\beta}n\big)
\Big)}.
\]
where $\Psi_n(b)$ is the minimizer of \eqref{eq::selector_KL_bound}. Hence if $I(W;J| \tilde{Z})=o(n)$, for every fixed
$\beta<\alpha$, 
\[
\mathbf P(F_n>\beta n)\ge q-o(1).
\]
\end{theorem}
The details are deferred to Appendix~\ref{app::details_mem}. 
Theorem~\ref{thm::mem} states that if an adversary declares at least $\alpha n$ training samples as ``in'' with probability at least $q$,then in the low-CMI regime it
must also declare more than $\beta n$ fresh samples as ``in'' with probability
asymptotically at least $q$, for every $\beta<\alpha$.

This improves \cite[Theorem~5(ii)]{sefidgaran2025tighter} that gives
$\mathbf P(F_n\ge m_\epsilon)\ge(\alpha-\epsilon)q$ for
$m_\epsilon=\frac{\epsilon}{1/q+\epsilon-\alpha}n-o(n)$. 
Since the leading term of $m_\epsilon$ is smaller than $\alpha$, we may choose
$\beta$ with $m_\epsilon<\beta n$ asymptotically and our result implies
$\mathbf P(F_n\ge m_\epsilon)\ge q-o(1)$, improves the probability level $(\alpha-\epsilon)q$.

Additionally, we notice that above results are from the change of measure inequality based on the KL divergence. 
Due to the flexibility of our DPI framework, we can also consider other measures. 

\begin{corollary}
\label{cor::selector_renyi_threshold}
Fix $\tau\in[0,1/2)$, $b_n:=\lfloor \tau n\rfloor$, and define $C_\tau:=\log 2-h_2(\tau)>0$. 
For every $\beta>1$,
\begin{equation}
p_{b_n}
\le
q_{b_n}^{\frac{\beta-1}{\beta}}
\exp\big(
\big((\beta-1)/\beta\big)
D_\beta(P_0\Vert Q_0)
\big).
\label{eq::selector_renyi_threshold_explicit}
\end{equation}
\end{corollary}
It promises an \emph{exponential} decay in the sample size $n$, see discussion and proof in Appendix~\ref{subapp::proof_selector_renyi_threshold}.

Also, by the differential privacy result (Lemma~\ref{lem:approxmaxinfo_Egamma_equiv}), we can derive a bound in terms of $I_\infty^{\tau}(J;Y)$:
\begin{proposition}
\label{prop::selector_recovery_approx_max_info} 
If $I_\infty^\tau(J;  Y) \leq \log \gamma$ for $\gamma   \geq \hspace{-1pt} 1$ and $\tau  \in [0,1)$, for every $\hat{J} =  \phi(Y)$ and $b\in\mathbb{N}_{\geq 0}$, 
\begin{equation}
p_b \leq \gamma q_b+\tau.
\label{eq::selector_recovery_approx_max_info}
\end{equation}
\end{proposition}
This connects the approximate max-information~\cite{rogers2016max} and differential privacy results to memorization, and can be viewed as a complementary perspective to the CMI-based anti-memorization results.

\section{Concluding Remarks and Limitations}
We have employed the  DPI to derive change of measure inequalities, a framework that is surprisingly elementary yet powerful enough. 
It is of interest to explore whether our framework can be connected to other unified approaches, e.g., compressibility bounds~\cite{sefidgaran2022rate} and a decorrelation lemma~\cite{chu2023unified}.
Moreover, the strong DPI has recently found applications in DNN generalization~\cite{he2025information} and data memorization~\cite{feldman2025trade}. 
Incorporating strong DPI into our framework is yet another future direction. 

Our approach requires an appropriate choice of $P$, $Q$, and $E$, which is not always straightforward across applications. 
While our DPI framework is often tight at the change-of-measure level, translating it into explicit generalization or memorization guarantees may still involve problem-dependent relaxations and require substantial effort. 
The practical implications of our generalization bounds are not yet fully clear, such as which bounds are more favorable in different settings, and the tightness of our bounds for practical algorithms requires further numerical investigation.

\section{Acknowledgement}

The authors are grateful to Professors Amin Gohari, Chandra Nair, and Cheuk Ting Li for their valuable comments on the first preprint of this paper. 
The authors would like to thank Mr. Muhan Guan for pointing out several mistakes in the PAC-Bayesian part of an earlier version. 
They also thank Dr. Gholamali Aminian for helpful comments on the same preprint. 
\newpage

\bibliographystyle{IEEEtran}
\bibliography{ref.bib}

\appendices

\section{More on Related Work}
\label{app::literature_review}

We provide a more detailed discussion of related works as follows.

\paragraph{Change of Measure Inequalities.}
Change of measure inequalities have been studied in learning theory since~\cite{csiszar1975divergence, donsker2006large} for PAC-Bayesian bounds.
In particular, \cite{hellstrom2020generalization} used a special change of measure inequality, namely the strong converse lemma~\eqref{eq::strong_converse} (see also~\cite{polyanskiy2025information}), to obtain various types of generalization error bounds. 
This lemma~\eqref{eq::strong_converse} has also been used widely for large-deviation analysis and hypothesis testing~\cite{polyanskiy2025information}. 
See also the use of change of measure inequalities in marginal structural models~\cite{roysland2011martingale} and robust uncertainty quantification bounds for statistical estimators~\cite{katsoulakis2017scalable}. 
Among change of measure inequalities, the two works that are most related to our results are~\cite{picard2022change} and~\cite{ohnishi2021novel}, which used the Young-Fenchel inequality and the variational representation of $f$-divergences, respectively, to derive new change of measure inequalities and hence new PAC-Bayesian generalization bounds; see also the references therein.

\paragraph{Generalization Error Bounds.}
Generalization error analysis has been one of the most important problems in machine learning in the past decades. It measures how a stochastic learning algorithm performs on data that are outside the training dataset, and there has been a vast array of work characterizing it. 
Early discussions include sample compression schemes~\cite{littlestone1986relating} (the output has to be ``compressible''; see also~\cite{vapnik1998statistical, boucheron2005theory}), model over-parameterization~\cite{shalev2014understanding} (over-parameterized models could overfit, although counterexamples were found by~\cite{zhang2017rethinking}), uniform stability~\cite{bousquet2002stability}, and the theory of uniform convergence~\cite{vapnik2015uniform} (the output has to be ``sufficiently simple''; see also~\cite{blumer1987occam}). 
However, it has been argued by~\cite{zhang2017rethinking} that algorithm-independent generalization bounds fall short of explaining the surprising success of modern artificial intelligence, and hence algorithm-dependent generalization bounds through different lenses have been explored.

In this direction, information-theoretic measures have been used to characterize generalization performance in the past decade. This line of work was initiated by~\cite{russo2016controlling, xu2017information}, who connected the generalization error to the \emph{mutual information} between the training dataset $S$ and the algorithm output $W$, indicating that algorithms that leak little information about the dataset generalize well (see also~\cite{banerjee2021information, bassily2018learners}). These bounds were further tightened using chaining~\cite{asadi2018chaining} and other techniques~\cite{harutyunyan2021information, bu2020tightening}. 
Since mutual information can easily be infinite, \emph{conditional} mutual information has been used by~\cite{steinke2020reasoning, haghifam2020sharpened, haghifam2021towards}. Viewing mutual information as a KL divergence, it was further generalized to more general information measures through $f$-divergences~\cite{esposito2021generalization, masiha2023f, wang2024generalization}. 
Various unified frameworks that cover multiple types of information measures have also been proposed, e.g., via information density~\cite{hellstrom2020generalization}, rate-distortion theory~\cite{sefidgaran2022rate, sefidgaran2024data}, convexity of the information measures~\cite{aminian2022tighter}, auxiliary distributions~\cite{aminian2024learning}, and a probabilistic decorrelation lemma~\cite{chu2023unified}.

\paragraph{Information Measures.}
In this paper, we explore the use of a wide range of information measures to characterize the generalization of stochastic algorithms. 
Among them, our novel results are mainly based on $f$-divergences, which were introduced by~\cite{ali1966general, csiszar1963informationstheoretische, csiszar1967information} as a generalization of the relative entropy that preserves useful properties, e.g., the data processing inequality~\cite{zakai1975generalization}. 
See~\cite{sason2016f} for a comprehensive study. 
A measure that is closely related to $f$-divergences is the R\'enyi divergence~\cite{renyi1961measures}, which has several useful operational interpretations, e.g., the number of bits by which a mixture of two codes can be compressed~\cite{harremoes2006interpretations, grunwald2007minimum} and the cutoff rate in block coding and hypothesis testing~\cite{csiszar2002generalized}. 
See also~\cite{van2014renyi} for a comprehensive study.

Besides $f$-divergence and R\'enyi divergence, another important class of information measures is the Sibson $\alpha$-mutual information~\cite{sibson1969information}; see also~\cite{verdu2015alpha} for a revisit. 
It generalizes both mutual information and maximal leakage~\cite{issa2016operational}. 
Maximal leakage, in particular, was also proposed as an information leakage measure, and has found recent applications in security; see~\cite{liao2017hypothesis, liao2018tunable, saeidian2023pointwise}. 
Variational representations of the Sibson $\alpha$-mutual information have been studied by~\cite{esposito2025sibson}. 
In~\cite{esposito2021generalization}, high-probability generalization bounds in terms of maximal leakage and Sibson $\alpha$-mutual information were proposed.
Also see the study of the generalization behavior of iterative, noisy learning algorithms~\cite{PensiaJogLoh2018ISIT} in terms of maximal leakage by~\cite{issa2023generalization}.

\paragraph{Privacy.}
Information leakage, or privacy, is tightly related to the generalization of stochastic algorithms, in the sense that private algorithms leak less information and are more ``stable''~\cite{dwork2015generalization, dwork2015preserving}. 
Differential privacy~\cite{dwork2006calibrating, dwork2014algorithmic} has been one of the most popular privacy measures in the past decade. 
Typicality, as a celebrated tool in information theory, has been used to derive generalization bounds for differentially private (and its variants~\cite{dong2022gaussian}) algorithms that are easy to compute~\cite{rodriguez2021upper, liu2026generalization}.

Interestingly, one information measure that we discussed throughout this paper, the $E_\gamma$-divergence, is tightly related to differential privacy. 
It is sometimes referred to as the hockey-stick divergence~\cite{sharma2013fundamental, sharma2012strong}; see its use in channel coding~\cite{polyanskiy2010channel, polyanskiy2010arimoto} and channel resolvability~\cite{liu2016e_, liu2015resolvability}. For the latter application, channel simulation~\cite{li2018strong}, a novel technique that is tightly related to channel resolvability, has been bridged to differential privacy as well~\cite{liu2024universal}. 
The contraction coefficients in the strong data processing inequality for the $E_\gamma$-divergence have been derived by~\cite{asoodeh2020contraction}, and it has been shown by~\cite{asoodeh2021local} that local differential privacy of a randomized algorithm can be equivalently cast in terms of the contraction of the $E_\gamma$-divergence.

\paragraph{Data Memorization.}
Data memorization has received substantial recent attention, both as a possible explanation of good \emph{generalization} in hard regimes and as a \emph{privacy} risk. 
Modern machine learning models can leak verbatim or near-verbatim information about their training sets, and this behavior has been connected to membership inference, extraction, and reconstruction attacks~\cite{carlini2019secret, carlini2021extracting, carlini2022quantifying, haim2022reconstructing}. 
On the theoretical side, \cite{feldman2020does} argued that long-tailed distributions can make memorization statistically useful, while \cite{brown2021when} identified natural prediction problems in which every sufficiently accurate learner must encode essentially all the information contained in a large subset of its training examples, including portions that are irrelevant to the target task. 
Related empirical studies further investigated which samples are memorized and how such memorization may decay during training~\cite{feldman2020what, jagielski2023measuring}.

In the context of stochastic convex optimization (SCO)~\cite{attias2024information}, memorization was quantified by conditional mutual information (CMI), and a trade-off between accuracy and information leakage was provided. 
A broader lower-bound framework was later developed by \cite{feldman2025trade}, who introduced \emph{excess data memorization} and related its lower bounds to strong data processing inequalities~\cite{polyanskiy2025information}. 
This yields sample-size/memorization trade-offs for binary prediction problems with latent structure, extending~\cite{brown2021when}. 
Because CMI is always at most $n$, these SDPI-based lower bounds are conceptually distinct from the CMI lower bounds of \cite{attias2024information}, and they emphasize a complementary phenomenon: additional data can gradually reduce the memorization burden. See also the follow-up work of \cite{voitovych2025traceability}. 
Meanwhile, \cite{sefidgaran2025tighter} showed that the necessity of memorization should not be interpreted too broadly. By combining stochastic projection with lossy compression and quantization, they derived tighter CMI-based generalization bounds that remain non-vacuous on the SCO counterexamples where standard MI/CMI bounds fail. 
These works suggest that data memorization is highly notion- and representation-dependent: it can be unavoidable for certain outputs or learner classes, yet avoidable after an appropriate randomized compression of the learned model.

\section{Proofs of Proposition~\ref{prop::E_gamma_converse}}
\label{app::E_gamma_converse}

A proof sketch has already been presented following Proposition~\ref{prop::E_gamma_converse} in Section~\ref{sec::change_mea_ineq}. 

In this section, we present two complete proofs of Proposition~\ref{prop::E_gamma_converse}. 
The first is a more detailed, complete version of the proof sketch that used \eqref{eq::DPI_f_div_detail}, highlighting the usefulness of our DPI approach. 
The second follows a more standard route and helps clarify why Proposition~\ref{prop::E_gamma_converse} is tighter than he strong converse lemma~\eqref{eq::strong_converse}.

\subsection{Proof 1 of Proposition~\ref{prop::E_gamma_converse}}

Consider $f(t) := [t-\gamma]_+$ that is convex on $(0,\infty)$ for any $\gamma\in\mathbb{R}$.
Recall $P\ll Q$.

If $q=0$, then by absolute continuity, $P(E)=0$ and \eqref{eq::Egamma_change_measure} holds trivially.

If $q=1$, then $p = 1$ and we have
$E_\gamma(P\Vert Q)\geq [1-\gamma]_+$,  hence
\[
\gamma+E_\gamma(P\Vert Q)\geq \gamma+[1-\gamma]_+\geq 1=P(E). 
\]

If $q\in(0,1)$, by data processing applied to the indicator map $\mathds{1}_E$,
\[
E_\gamma(P\Vert Q)=D_f(P\Vert Q)\geq D_f(\mathds{1}_E\circ P \Vert \mathds{1}_E\circ Q)
= D_f(\mathrm{Ber}(p) \Vert \mathrm{Ber}(q)).
\]

With $f(t) := [t-\gamma]_+$, we calculate 
\[
D_f(\mathrm{Ber}(p) \Vert \mathrm{Ber}(q))
= q f \Big(\frac{p}{q}\Big) + (1-q)  f \Big(\frac{1-p}{1-q}\Big)
= q\Big[\frac{p}{q} - \gamma\Big]_+ + (1-q)\Big[\frac{1-p}{1-q}-\gamma\Big]_+.
\]

Using $[x]_+\geq x$ for all $x$ and $(1-q)[\cdot]_+\geq 0$, we get
\[
E_\gamma(P\Vert Q)\geq q\Big(\frac{p}{q}-\gamma\Big)=p-\gamma q, 
\]
and rearrangement yields $p\leq \gamma q + E_\gamma(P\Vert Q)$, i.e. \eqref{eq::Egamma_change_measure}.

\subsection{Proof 2 of Proposition~\ref{prop::E_gamma_converse}}
\label{subsec::proof2_Egamma}

For probability distributions $P$ and $Q$ such that $P\ll Q$, for any measurable set $E$ and any $\gamma \in \mathbb{R}$, we have 
\begin{align*}
P(E) 
&= \int \mathds{1}_E \mathrm{d}P\\
&= \int \mathds{1}_E \frac{\mathrm{d}P}{\mathrm{d}Q} \mathrm{d}Q\\
&= \int \mathds{1}_E 
\left(
\min\left\{ \frac{\mathrm{d}P}{\mathrm{d}Q}, \gamma \right\} + \left[ \frac{\mathrm{d}P}{\mathrm{d}Q} - \gamma \right]_+
\right)
\mathrm{d}Q\\
&= \int \mathds{1}_E 
\left(
\min\left\{ \frac{\mathrm{d}P}{\mathrm{d}Q}, \gamma \right\}\right)  \mathrm{d}Q + 
\int \mathds{1}_E  \left( \left[ \frac{\mathrm{d}P}{\mathrm{d}Q} - \gamma \right]_+
\right)
\mathrm{d}Q\\
&\leq \gamma Q(E) + \int \mathds{1}_E  \left( \left[ \frac{\mathrm{d}P}{\mathrm{d}Q} - \gamma \right]_+
\right)
\mathrm{d}Q\\
&\leq \gamma Q(E) + 
\int \left( \left[ \frac{\mathrm{d}P}{\mathrm{d}Q} - \gamma \right]_+
\right)
\mathrm{d}Q\\
& = \gamma Q(E) + 
E_{\gamma}(P\Vert Q). 
\end{align*}

\subsection{Discussions on Strong Converse Lemma}
\label{app::proof_strong_converse}

In this subsection, we discuss the strong converse lemma~\cite[Chapter~14]{polyanskiy2025information}, and in particular why Proposition~\ref{prop::E_gamma_converse} turns out to be tighter.

We first review the strong converse lemma~\eqref{eq::strong_converse} together with its proof, which assist our discussion later. 

\begin{lemma}
    \label{lem::strong_converse}
    For probability distributions $P$ and $Q$ such that $P\ll Q$, for any measurable set $E$ and any $\gamma \in \mathbb{R}$, we have 
    \[
    P(E)\leq \gamma Q(E) + P\left( \frac{\mathrm{d}P}{\mathrm{d}Q}>\gamma\right).
    \]
\end{lemma}

\begin{proof}
    For probability distributions $P$ and $Q$ such that $P$ is absolutely continuous with respect to $Q$, for a measurable set $E$ and any $\gamma \in\mathbb{R}$, we have 
    \begin{align*}
    P(E) 
    &= \int \mathds{1}_E \mathrm{d}P\\
    &= \int \mathds{1}_E\mathds{1}\left(\frac{\mathrm{d}P}{\mathrm{d}Q}\leq \gamma\right)\left(\frac{\mathrm{d}P}{\mathrm{d}Q}\right)\mathrm{d}Q + \int \mathds{1}_E\mathds{1}\left(\frac{\mathrm{d}P}{\mathrm{d}Q}> \gamma\right)\mathrm{d}P\\
    &\leq \int \mathds{1}_E \gamma \mathrm{d}Q + \int \mathds{1}_E\mathds{1}\left(\frac{\mathrm{d}P}{\mathrm{d}Q}> \gamma\right)\mathrm{d}P\\
    &=\gamma Q(E) + \int 
    \mathds{1}_E
    \mathds{1}\left(\frac{\mathrm{d}P}{\mathrm{d}Q}> \gamma\right)
    \mathrm{d}P\\
    &\leq \gamma Q(E) + \int \mathds{1}\left(\frac{\mathrm{d}P}{\mathrm{d}Q}> \gamma\right)\mathrm{d}P\\
    &= \gamma Q(E) + P\left( \frac{\mathrm{d}P}{\mathrm{d}Q}>\gamma \right).
\end{align*}
\end{proof}

Proposition~\ref{prop::E_gamma_converse} is tighter than~Lemma~\ref{lem::strong_converse}, and it exactly recovers Lemma~\ref{lem::strong_converse} by considering
\begin{equation*}
    E_{\gamma}(P\Vert Q) \leq\int_{\big\{\frac{\mathrm{d}P}{\mathrm{d}Q} > \gamma \big\}} \frac{\mathrm{d}P}{\mathrm{d}Q} \mathrm{d}Q 
    = P\left(\Big\{\frac{\mathrm{d}P}{\mathrm{d}Q} > \gamma \Big\}\right). 
\end{equation*}

It is easy to compare the second proof of Proposition~\ref{prop::E_gamma_converse}, shown in Section~\ref{subsec::proof2_Egamma}, to the proof of the strong converse lemma as shown above. 
The advantage of Proposition~\ref{prop::E_gamma_converse} can be understood as follows. 
The event $\{\frac{\mathrm{d}P}{\mathrm{d}Q} > \gamma\}$ is the region where the event on $P$ is much heavier than that on $Q$; to bound $P(E)$ we essentially need a pointwise inequality such that \begin{equation}
\label{eq::r_omega_requirement}
\frac{\mathrm{d}P}{\mathrm{d}Q}(\omega) \leq \gamma + r(\omega)
\end{equation}
and then integrate $r$ on $Q$. To prove Proposition~\ref{prop::E_gamma_converse}, we employ \[
r(\omega) = \left[\frac{\mathrm{d}P}{\mathrm{d}Q}(\omega) - \gamma\right]_+,
\]
which is the smallest nonnegative function satisfying \eqref{eq::r_omega_requirement}, while the choice of $r(\omega)$ in the proof of Lemma~\ref{lem::strong_converse} is \[
\frac{\mathrm{d}P}{\mathrm{d}Q}(\omega) \, 
\mathds{1} \left(\frac{\mathrm{d}P}{\mathrm{d}Q}(\omega) > \gamma\right),
\]
which is lower bounded by $[\mathrm{d}P/\mathrm{d}Q(w)-\gamma]_+$. 

From another perspective, the proof of Lemma~\ref{lem::strong_converse} effectively discards a conditional term of the form $P(A|E)$, thereby losing information about how the event $E$ is positioned relative to the ``high-information'' region $\{\frac{\mathrm{d}P}{\mathrm{d}Q}>\gamma\}$. 
In contrast, the $E_\gamma$-divergence captures both \emph{how often} $\frac{\mathrm{d}P}{\mathrm{d}Q}$ exceeds $\gamma$ and \emph{how far above} $\gamma$ it typically lies. 
Consequently, $E_\gamma(P\Vert Q)$ retains finer information about ``where $E$ sits'' within the high-information regime, partially compensating for what is lost by dropping $P(A|E)$.

\section{Proof of Theorem~\ref{thm::hellinger_chi2} and Discussions}
\label{sec::proof_prop_hellinger_chi2}

In this section, we present the proof of Theorem~\ref{thm::hellinger_chi2} case by case, all of which are based on~\eqref{eq::DPI_f_div_detail}. Some cases require further fine relaxations, which will also be discussed in Section~\ref{subsec::comparison} together with bound comparisons.

\subsection{Proof of~\eqref{eq::chi2}}

Consider the $\chi^2$ distance is an $f$-divergence with $f = t^2-1$.

By~\eqref{eq::DPI_f_div_detail} we can calculate
\[
    \chi^2(P\Vert Q) \geq Q(E)\cdot \left(\left(\frac{P(E)}{Q(E)}\right)^2-1\right) + (1-Q(E))\cdot \left(\left(\frac{1-P(E)}{1-Q(E)}\right)^2-1\right),
\]
and rearrangement yields~\eqref{eq::chi2}.

\subsection{Proof of~\eqref{eq::PQ_KL}}

Consider the KL divergence is an $f$-divergence with $f(t) = t\log t$, at first it is easy to find 
\begin{align*}
    \mathrm{KL}(P\Vert Q) &\geq P(E)\log \frac{P(E)}{Q(E)} + (1-P(E))\log \frac{1-P(E)}{1-Q(E)}\\
    &=P(E)\log \frac{1}{Q(E)} + (1-P(E))\log \frac{1}{1-Q(E)} - h_2(P(E))\\
    &\geq P(E)\log \frac{1}{Q(E)} - \log 2.
\end{align*}
However, the last inequality can be a bit crude so that it is not always tighter than~\cite{picard2022change}. To improve the bound, we can perform a finer analysis as follows. 

Let $p:=P(E)$ and $q:=Q(E)\in(0,1)$, we can follow the same procedure until
\begin{align*}
\mathrm{KL}(P\Vert Q) \geq  p\log \frac{1}{q} + (1-p)\log \frac{1}{1-q} - h_2(p),
\end{align*}
and then we instead use the Fenchel inequality for $h_2$: for every $t\in\mathbb{R}$,
\begin{equation}
\label{eq:fenchel_entropy}
h_2(p)\le \log \big(1+e^t\big)-pt.
\end{equation}
Substituting \eqref{eq:fenchel_entropy} into the previous display yields, for every $t\in\mathbb{R}$,
\begin{align*}
\mathrm{KL}(P\Vert Q)
&\ge p\log \frac{1}{q} + (1-p)\log \frac{1}{1-q} + pt-\log(1+e^t)\\
&= p\Big(t+\log\frac{1-q}{q}\Big) - \log(1-q) - \log(1+e^t).
\end{align*}
If we take $c := t+\log\frac{1-q}{q}$, we have 
\[
\log(1+e^t)
=\log \Big(1+e^{c}\frac{q}{1-q}\Big)
=\log(1-q+qe^c)-\log(1-q),
\]
and hence, for every $c\in\mathbb{R}$,
\begin{equation}\label{eq:kl_event_affine}
\mathrm{KL}(P\Vert Q) \geq p c-\log(1-q+qe^c).
\end{equation}
By rearrangement we have for $c>0$, 
\begin{equation}
\label{eq:kl_event_improved_A}
P(E)
\leq \frac{\mathrm{KL}(P\Vert Q)+\log \big(1+q(e^{c}-1)\big)}{c}.
\end{equation}

\subsection{Proof of~\eqref{eq::PQ_Hellinger_distance}}

Consider the Hellinger squared distance is an $f$-divergence with $f = (\sqrt{t}-1)^2$. 

By~\eqref{eq::DPI_f_div_detail} we can calculate
\begin{align*}
    H^2(P;Q) &\geq Q(E)\cdot \left(\sqrt{\frac{P(E)}{Q(E)}} - 1\right)^2 + (1-Q(E))\cdot \left(\sqrt{\frac{1 - P(E)}{1 - Q(E)}} - 1\right)^2\\
    &= 2\left(1 - \sqrt{P(E)Q(E)} - \sqrt{(1-P(E))(1-Q(E))}\right), 
\end{align*}
which results in~\eqref{eq::PQ_Hellinger_distance}.

\subsection{Proof of~\eqref{eq::PQ_Hellinger_divergence2}}

Consider the Power divergence of order $\beta$ is an $f$-divergence with $f = (t^\beta - 1)/(\beta - 1)$. 

By~\eqref{eq::DPI_f_div_detail} we can calculate
\[
    \mathcal{H}_\beta(P\Vert Q) \geq \frac{1}{\beta - 1}\cdot \left(Q(E)\left(\left(\frac{P(E)}{Q(E)}\right)^\beta - 1\right) + (1-Q(E))\left(\left(\frac{1-P(E)}{1-Q(E)}\right)^\beta - 1\right)\right),
\]
and rearrangement yields~\eqref{eq::PQ_Hellinger_divergence2}.

\subsection{Further Relaxations}
\label{subapp::further_relax}

To convert \eqref{eq::PQ_Hellinger_distance} and \eqref{eq::PQ_Hellinger_divergence2} to the form of~\eqref{eq::change of measure}, we can relax them as follows: 
\begin{align}
    P(E) &\leq \Big(\sqrt{Q(E)\left(1-H^2(P;Q)/2\right)^2}+\sqrt{(1-Q(E))\big(1-\left(1-H^2(P;Q)/2\right)^2\big)}\,\,\Big)^2,\label{eq::PQ_Hellinger_distance_relaxed}\\
    P(E) & \leq ((\beta - 1)\mathcal{H}_\beta(P\Vert Q) + 1)^{1/\beta}\big/ \big(Q(E)^{1-\beta} + M^\beta(1-Q(E))^{1-\beta} \big)^{1/\beta},\label{eq::PQ_Hellinger_divergence2_relaxed}  
\end{align}
where $M\leq 1/P(E)-1$ in~\eqref{eq::PQ_Hellinger_divergence2_relaxed}, which is a valid upper bound on $P(E)$ as long as $Q(E)$ is bounded by $Q_{\max}<1$. 
For such a case with fixed $\beta$, we let $M = Q_{\max}^{{(1-\beta)}/{\beta}}\cdot ((\beta-1)\mathcal{H}_\beta(P\Vert Q)+1)^{-1/\beta} - 1$.

We can relax~\eqref{eq::PQ_Hellinger_divergence2} to  derive a bound that is even tighter than~\eqref{eq::PQ_Hellinger_divergence2_relaxed} when $Q(E)$ is small, precisely the regime relevant for generalization bounds where $Q(E)$ decays exponentially with sample size\footnote{If $X\sim P_X$ is $\sigma$-sub-Gaussian, Hoeffding's inequality gives $\forall \eta > 0$, $P_X\left(|X - \mathbf{E}[X]|\geq \eta \right)\leq 2 \exp ( - {\eta^2}/{2 \sigma^2} )$.}: 
\begin{equation}
P(E)\leq \big (
(Q(E))^{\beta-1}\big[1 + (\beta-1)\mathcal H_\beta(P\Vert Q) - (1-Q(E))^{1-\beta}(1-u_0)^\beta\,\big]_+
\big)^{1/\beta},\label{eq::PQ_Hellinger_divergence2_relaxed2} 
\end{equation}
where $u_0 := \min\big\{1,\ \big((1+(\beta-1)\mathcal H_\beta(P\Vert Q))\,q^{\beta-1}\big)^{1/\beta}\big\}$. See Appendix~\ref{app::PQ_Hellinger_divergence2_relaxed2} for a proof.

Based on~\eqref{eq::Hellinger_to_Renyi}, one can also convert~\eqref{eq::PQ_Hellinger_divergence2_relaxed} and~\eqref{eq::PQ_Hellinger_divergence2_relaxed2} to inequalities in terms of R\'enyi divergence. 

We then compare our bounds with existing results. 
It turns out that, in almost all cases, our bounds are tighter than the existing  results~\cite{picard2022change, esposito2021generalization}, while~\eqref{eq::chi2} and~\eqref{eq::PQ_KL} recover the results in~\cite{picard2022change}; we conjecture that they may already be optimal. 
The bound in \eqref{eq::chi2} improves upon~\cite[Corollary~7]{esposito2021generalization} whenever $\chi^2(P\Vert Q)\lesssim {1}/{(4Q(E))}$ and \eqref{eq::PQ_KL} is tighter than~\cite[Lemma~9]{bassily2018learners}. 
The bound in \eqref{eq::PQ_Hellinger_distance} is tighter than~\cite[Corollary~9]{esposito2021generalization} and does not require $P(E)\geq Q(E)$ as they do. 
The bound in \eqref{eq::PQ_Hellinger_divergence2} recovers~\cite[Corollary~7]{esposito2021generalization} by dropping the second term on the left-hand side. 
Finally, the bound in \eqref{eq::PQ_Hellinger_divergence2_relaxed2} is tighter than~\cite[Corollary 7]{esposito2021generalization} whenever $Q(E)$ is small. Comparing~\eqref{eq::PQ_Hellinger_divergence2_relaxed2} to~\cite{picard2022change}, neither is stronger than another, but one advantage of~\eqref{eq::PQ_Hellinger_divergence2_relaxed2} it does not need to optimize over $s\in\mathbb{R}$.

\subsection{Proof of \eqref{eq::PQ_Hellinger_divergence2_relaxed2}}
\label{app::PQ_Hellinger_divergence2_relaxed2}

We let $p := P(E)$ and $q := Q(E)$. Recall $u_0 := \min\big\{1,\ \big((1+(\beta-1)\mathcal H_\beta(P\Vert Q))\,q^{\beta-1}\big)^{1/\beta}\big\}$, which is less than $1$ when $Q(E)$ is small. 

From \eqref{eq::PQ_Hellinger_divergence2} and $(1-p)^\beta(1-q)^{1-\beta}\ge 0$, we get
$p^\beta q^{1-\beta}\le 1+(\beta-1)\mathcal H_\beta(P\Vert Q)$, hence $p\le u_0$.
Therefore $1-p\ge 1-u_0$, which implies $(1-p)^\beta \ge (1-u_0)^\beta$.
Plugging this into \eqref{eq::PQ_Hellinger_divergence2} yields
\[
p^\beta q^{1-\beta} 
\le 1+(\beta-1)\mathcal H_\beta(P\Vert Q) - (1-u_0)^\beta(1-q)^{1-\beta}.
\]
Multiplying $q^{\beta-1}$ and then taking the $1/\beta$ power on both sides give the desired result.

\subsection{Discussions and Comparison}
\label{subsec::comparison}

We then discuss our inequalities and compare them with existing results.

\paragraph{KL divergence.}

From~\cite[Lemma~9]{bassily2018learners} we have 
\begin{equation}\label{eq:kl_event_crude_B}
P(E)\le \frac{\mathrm{KL}(P\Vert Q)+\log 2}{\log(1/Q(E))}.
\end{equation}

It is easy to see \eqref{eq:kl_event_improved_A} implies \eqref{eq:kl_event_crude_B} by 
\[
\log \big(1+q(e^c-1)\big)\le \log(1+e^c)\le c+\log 2,
\]
and in particular, choosing $c=\log(1/q)>0$ in \eqref{eq:kl_event_improved_A} yields 
\begin{equation}\label{eq:kl_event_improved_specialize}
P(E)\le \frac{\mathrm{KL}(P\Vert Q)+\log(2-q)}{\log(1/q)}
< 
\frac{\mathrm{KL}(P\Vert Q)+\log 2}{\log(1/q)},
\end{equation}
which is strict since $2-q<2$ for every $q\in(0,1)$.

\paragraph{Power divergence.}
The power divergence bound~\eqref{eq::PQ_Hellinger_divergence2} is equivalent to 
\[
\mathcal{H}_\beta(P\Vert Q) \geq 
\frac{P(E)^\beta Q(E)^{1-\beta} + (1-P(E))^\beta(1-Q(E))^{1-\beta} - 1}{\beta - 1}.
\]
One crude but clean relaxation is \[
\mathcal{H}_\beta(P\Vert Q) \geq \frac{P(E)^\beta Q(E)^{1-\beta} - 1}{\beta - 1},
\]
implying the result \[
P(E)\leq Q(E)^{\frac{\beta - 1}{\beta}}\left((\beta - 1)\mathcal{H}_\beta(P\Vert Q) + 1\right)^{\frac{1}{\beta}},
\]
which exactly recovers~\cite[Corollary 7]{esposito2021generalization}.

If we assume $P(E)\leq \frac{1}{M+1}$ for some $M>0$, which is reasonable in generalization error analysis since the error is usually not large, then we have $1-P(E)\geq M P(E)$, and hence obtain~\eqref{eq::PQ_Hellinger_divergence2_relaxed}: 
\[
P(E)\leq \frac{((\beta - 1)\mathcal{H}_\beta(P\Vert Q) + 1)^{\frac{1}{\beta}}}{(Q(E)^{1-\beta}+M^\beta(1-Q(E))^{1-\beta})^{\frac{1}{\beta}}}.
\]

We here verify the validity of the assumption on $P(E)$. By fixing $\beta$ and assuming that $Q(E)$ is upper bounded by some $Q_{\max}<1$, we can first upper bound $P(E)$ by 
\[
P(E) \leq Q_{\max}^{\frac{\beta-1}{\beta}}\cdot ((\beta-1)\mathcal{H}_{\beta}(P\Vert Q)+1)^{1/\beta}=\frac{1}{M+1},
\]
which is equivalent to 
\[
M = Q_{\max}^{\frac{1-\beta}{\beta}}\cdot ((\beta-1)\mathcal{H}_\beta(P\Vert Q)+1)^{-1/\beta} - 1.
\]

However, we note that in the regime where generalization error analysis is performed, i.e., $Q(E)$ is small due to Hoeffding's inequality, above is usually not as tight as~\eqref{eq::PQ_Hellinger_divergence2_relaxed2}.

\paragraph{Hellinger Squared Distance.} 
The Hellinger squared distance bound~\eqref{eq::PQ_Hellinger_distance} can easily recover the same functional form as~\cite[Corollary~9]{esposito2021generalization}, without requiring the condition $P_{SW}(E)\geq P_S P_W(E)$, which is assumed in~\cite[Corollary~9]{esposito2021generalization}. 
By utilizing a better relaxation, \eqref{eq::PQ_Hellinger_divergence2_relaxed} is a tighter bound.

By~\cite{picard2022change}, we have
\[
P(E) \leq 1+c - \frac{c(1+c)(1-H^2(P;Q))^2}{Q(E)+c}
\]
for any $c>0$.
The optimal choice of $c$ is for them is 
\[
c^* = -Q(E) + \sqrt{\frac{(1-H^2(P;Q))^2 Q(E)(1 - Q(E))}{1-(1-H^2(P;Q))^2}}, 
\]
which gives 
\begin{equation}
    P(E) \leq \Big(\sqrt{\bigl(1-(1-H^2(P;Q))^2\bigr)\bigl(1-Q(E)\bigr)}+\bigl(1-H^2(P;Q)\bigr)\sqrt{Q(E)}\Big)^2.
    \label{eq::picard_PE_hellinger_dist}
\end{equation}

For our bound shown in~\eqref{eq::PQ_Hellinger_distance}, with $f(t)=(\sqrt{t}-1)^2$, we have \begin{align*}
    H^2(P;Q) &\geq Q(E)\cdot \left(\sqrt{\frac{P(E)}{Q(E)}} - 1\right)^2 + (1-Q(E))\cdot \left(\sqrt{\frac{1 - P(E)}{1 - Q(E)}} - 1\right)^2\\
    &= 2\left(1 - \sqrt{P(E)Q(E)} - \sqrt{(1-P(E))(1-Q(E))}\right),
\end{align*}
which is equivalent to \begin{align*}
1-\frac{H^2(P;Q)}{2} &\leq \sqrt{P(E)Q(E)} + \sqrt{(1-P(E))(1-Q(E))}.
\end{align*}
Then we can solve that \begin{equation*}
    P(E) \leq \left(\sqrt{Q(E)\left(1-\frac{H^2(P;Q)}{2}\right)^2}+\sqrt{(1-Q(E))\left(1-\left(1-\frac{H^2(P;Q)}{2}\right)^2\right)}\right)^2,
\end{equation*}
which is~\eqref{eq::PQ_Hellinger_distance_relaxed}.
Observe that both the above result and the optimal result in~\cite{picard2022change} can be written in the form of \[
F(x) = \left(\sqrt{Q(E)x^2}+\sqrt{(1-Q(E))(1-x^2)}\right)^2.
\]
This function $F(x)$ is increasing on $x\in(0,\sqrt{Q(E)})$ and decreasing on $x\in[\sqrt{Q(E)},1)$. Since \eqref{eq::picard_PE_hellinger_dist} is $F(1-H^2(P;Q))$ and our result \eqref{eq::PQ_Hellinger_distance_relaxed} is $F(1-H^2(P;Q)/2)$. For \[
\sqrt{Q(E)}\leq 1-H^2(P;Q)\leq 1-\frac{H^2(P;Q)}{2},
\]
we have \[
F(1-H^2(P;Q)) \geq F(1-H^2(P;Q)/2),
\]
implying that our result~\eqref{eq::PQ_Hellinger_distance_relaxed} is tighter than~\eqref{eq::picard_PE_hellinger_dist} given in~\cite{picard2022change}.



    Then we compare \eqref{eq::PQ_Hellinger_distance_relaxed} with the result in~\cite{esposito2021generalization}. Observe that~\eqref{eq::PQ_Hellinger_distance_relaxed} is equivalent to
    \begin{align*}
        P(E) 
        & \leq 
        \Bigg(\sqrt{Q(E)}\Big(1-\frac{H^{2}(P;Q)}{2}\Big)
        +\sqrt{1-Q(E)}\cdot\frac{\sqrt{H^{2}(P;Q)\big(4-H^{2}(P;Q)\big)}}{2}\Bigg)^{2}.
    \end{align*}
    Using $\sqrt{1-(1-H^2(P;Q)/2)^2} = \sqrt{H^{2}(P;Q) - (H^{2}(P;Q))^2/4} \leq \sqrt{H^{2}(P;Q)}$ we can relax above to 
    \[
    P(E) \leq \left(\sqrt{Q(E)}+\sqrt{H^{2}(P;Q)}\right)^{2}.
    \]

\section{Background and Proof of Theorem~\ref{thm::E_gamma_converse_generalized_Orlicz}}
\label{app::E_gamma_converse_generalized_Orlicz}

In this section, we provide technical background of Theorem~\ref{thm::E_gamma_converse_generalized_Orlicz}, as well as its proof. 

\subsection{Technical Background}

We first present necessary background and technical details for Theorem~\ref{thm::E_gamma_converse_generalized_Orlicz}. 
For more discussions, we refer the readers to~\cite{hudzik2000amemiya, jiao2017dependence, esposito2021generalization}.

We say a function $\psi$ is an Orlicz function if it is a convex function $\psi:[0,\infty) \rightarrow[0,\infty]$ that vanishes at zero and is not identically $0$ or $\infty$ on $(0, \infty)$. 
Given a convex function $\psi:[0,\infty) \rightarrow\mathbb{R}$, define its conjugate $\psi^\star :[0,\infty) \rightarrow\mathbb{R}$ as 
\begin{equation*}
    \psi^\star(t) = \sup_{\lambda>0} \lambda t - \psi(\lambda).
\end{equation*}
The generalized inverse of $\psi$ is defined as 
\begin{equation*}
    \psi^{-1}(s) :=\inf\big\{ t\geq 0: \psi(t)\geq s \big\}
\end{equation*}
for $s\geq 0$ and with the convention that $1/0 = \infty$ and $\psi^{-1}(\infty) = \infty$.

Consider a complete and $\sigma$-finite probability space $(\Omega,\mathcal{F}, \mu)$. 
Let $L^0(\mu)$ denote the space of all the $\mathcal{F}$-measurable and real valued functions on $\Omega$. 
We can define a functional $I_\psi: L^0(\mu)\rightarrow [0, \infty]$ as \[I_\psi(x) = \int_\Omega \psi(|x(t)|)\mathrm{d}\mu(t)\] and then an Orlicz space can be defined  as~\cite{hudzik2000amemiya}
\[
L_\psi(\mu) = \left\{
x\in L^0(\mu):I_\psi (\lambda x) < \infty\text{ for some } \lambda > 0
\right\}. 
\]
Note the Orlicz space is a Banach space and can be endowed with Luxemburg, Orlicz and Amemiya norms. 
It has been shown in~\cite{hudzik2000amemiya} that in general memiya norms are equivalent to Orlicz norms. 
Same to~\cite{esposito2021generalization}, in this paper we restrict our discussions to probability spaces and define the corresponding norms with respect to random variables and the expectation operator.

For other uses of duality in Orlicz spaces in the context of generalization bounds, we refer the readers to~\cite{esposito2021generalization, chu2023unified}.

Let $U$ be an $\mathcal{F}$-measurable random variable.
then the Luxemburg norm of $U$ with respect to $\mu$ is defined as
\begin{equation*}
    \Vert U\Vert_\psi^\mu = \inf\left\{
    \sigma>0:\mathbf{E}_\mu\left[\psi\Big(\frac{|U|}{\sigma} \Big) \right]\leq 1
    \right\}, 
\end{equation*}
and the Amemiya norm of $U$ with respect to $\mu$ is defined as
\begin{equation*}
    \Vert U\Vert_\psi^{A, \mu} = \inf\left\{
    \frac{\mathbf{E}_\mu\left[\psi\big(t|U| \big) \right] + 1}{t}: t>0
    \right\}. 
\end{equation*}

Then we introduce the following lemma:

\begin{lemma}[\cite{hudzik2000amemiya}]
    \label{lem::Orlicz_holder}
    Let $\psi$ be an Orlicz function and $\psi^\star$ denote its conjugate, then for every couple of random variable $U,V$, we have 
    \begin{equation*}
        \mathbf{E}[UV]\leq \Vert U\Vert_\psi \Vert V\Vert_{\psi^\star}^A.
    \end{equation*}
\end{lemma}
Note this Lemma~\ref{lem::Orlicz_holder} recovers the Hölder's inequality by taking $\psi(t) = t^\alpha/\alpha$ (which consequently gives $\psi^\star(t) = t^\beta/\beta$ with $1/\alpha + 1/\beta = 1$). 
One proof can be found in~\cite[Appendix A]{esposito2021generalization}.

\subsection{Proof of our Theorem~\ref{thm::E_gamma_converse_generalized_Orlicz}}

We then provide formal proof of our Theorem~\ref{thm::E_gamma_converse_generalized_Orlicz}. 

For probability distributions $P$ and $Q$ such that $P\ll Q$, for a measurable set $E$, fix $\gamma\in\mathbb{R}$, and let $\psi$ be an Orlicz function
with convex conjugate $\psi^\star$. 
Recall we denote by $\Vert\cdot\Vert_{\psi}^{Q}$ the Luxemburg norm and by $\Vert\cdot\Vert_{\psi^\star}^{A,Q}$ the Amemiya norm.

For a measurable set $E$, we have
\begin{align*}
P(E)
&= \int \mathds{1}_{E} \,\mathrm{d}P\\
&= \int \mathds{1}_{E}\, \frac{\mathrm{d}P}{\mathrm{d}Q} \,\mathrm{d}Q\\
&= \int \mathds{1}_{E}\Bigg(\min\left\{\frac{\mathrm{d}P}{\mathrm{d}Q},\gamma\right\}
      + \left[\frac{\mathrm{d}P}{\mathrm{d}Q}-\gamma\right]_+\Bigg)\,\mathrm{d}Q\\
&= \int \mathds{1}_{E}\min\left\{\frac{\mathrm{d}P}{\mathrm{d}Q},\gamma\right\}\,\mathrm{d}Q
    +   \int \mathds{1}_{E}\, \left[\frac{\mathrm{d}P}{\mathrm{d}Q}-\gamma\right]_+\,\mathrm{d}Q\\
&\stackrel{(a)}{\le} \gamma\,Q(E)
    +   \int \mathds{1}_{E}\, \left[\frac{\mathrm{d}P}{\mathrm{d}Q}-\gamma\right]_+\,\mathrm{d}Q\\
&\stackrel{(b)}{\le} \gamma\,Q(E)
  +   \big\Vert\mathds{1}_{E}\big\Vert_{\psi}^{Q}\,
      \left\Vert\left[\frac{\mathrm{d}P}{\mathrm{d}Q}-\gamma\right]_+\right\Vert_{\psi^\star}^{A,Q},
\end{align*}
where $(a)$ is because for every $s \in\mathcal{S}$, $0\leq \mathds{1}_E(x)\min\{\frac{\mathrm{d}P}{\mathrm{d}Q}(x), \gamma\}\leq \mathds{1}_E(x) \gamma$, and integrating both sides w.r.t $Q$ gives $\int \mathds{1}_E\min\{\frac{\mathrm{d}P}{\mathrm{d}Q}, \gamma\} \mathrm{d}Q \leq \int \mathds{1}_E(x) \gamma \mathrm{d}Q = \gamma Q(E)$; $(b)$ is by the generalized Hölder's inequality in Lemma~\ref{lem::Orlicz_holder}, which is applied to
$U=\mathds{1}_{E}$ and $V=\big[\frac{\mathrm{d}P}{\mathrm{d}Q}-\gamma\big]_+$.

This proves the Orlicz form
\begin{equation}
    P(E)\leq \gamma Q(E)
+ \big\Vert\mathds{1}_{E}\big\Vert_{\psi}^{Q}\,
  \left\Vert\left[\frac{\mathrm{d}P}{\mathrm{d}Q}-\gamma\right]_+\right\Vert_{\psi^\star}^{A,Q}.
\label{eq::Orlicz_generalized_Egamma}
\end{equation}

Then we note that the Luxemburg norm of the indicator has the explicit form 
\begin{equation*}
 \Vert \mathds{1}_{E}\Vert_\psi^Q = \frac{1}{\psi^{-1}(1/Q(E))}
\end{equation*}
where the generalized inverse is
\begin{equation*}
 \psi^{-1}(s) :=\inf\big\{
 t\geq 0: \psi(t)\geq s
 \big\}
\end{equation*}
for $s\geq 0$ and with the convention that $1/0 = \infty$ and $\psi^{-1}(\infty) = \infty$.

\subsection{Further Generalizations}

Moreover, if we employ H\"older's inequality, which is also used by~\cite{chu2023unified, esposito2021generalization}, we can further generalize Proposition~\ref{prop::E_gamma_converse}. 
Theorem~\ref{thm::E_gamma_converse_generalized_Orlicz} implies the following theorem. 

\begin{theorem}
\label{thm::gamma_refined_orlicz}
Assume $P_{SW}\ll P_SP_W$. Given $E\in \mathcal{F}$ and two Orlicz functions $\psi,\varphi$, for any $\gamma\in\mathbb{R}$,  
\begin{align}
P_{SW}(E)
 \leq 
\gamma\,P_SP_W(E)
  +  
\left\lVert \left\lVert \mathds{1}_{W\in E_W} \right\rVert^{P_S}_\varphi\right\rVert^{P_W}_\psi
\cdot
\left\lVert\left\lVert \left[\frac{\mathrm{d}P_{SW}}{\mathrm{d}P_SP_W}-\gamma\right]_+\right\rVert^{A,P_S}_{\varphi^\star}\right\rVert^{A,P_W}_{\psi^\star}, \label{eq:gamma_refined_orlicz_main}
\end{align}
where for each $w \in \mathcal{W}$, $E_w := \{ s: (s,w) \in E \}$ (\textit{i.e.}, the ``fiber'' of $E$ with respect to $w$).
\end{theorem}
The proof of Theorem~\ref{thm::gamma_refined_orlicz}  can be found in Appendix~\ref{app::proof_thm_gamma_refined_orlicz}. 
It strictly generalizes~\cite[Theorem 1]{esposito2021generalization} (and recovers it by taking $\gamma = 0$), which in turn implied various generalization error bounds. 
We strictly improve them when the reduction of the Amemiya term outweighs the additive $\gamma\,P_SP_W(E)$, usually in the scenarios where $P_S P_W (E)$ is tiny but $\mathrm{d}P_{SW} / \mathrm{d}P_SP_W$ has heavy tails. 

\subsection{Proof of Theorem~\ref{thm::gamma_refined_orlicz}}
\label{app::proof_thm_gamma_refined_orlicz}

\begin{proof}
Fix $\gamma\geq 0$, assume $P_{SW}\ll P_S P_W$ and denote $L(s,w):=\frac{\mathrm dP_{SW}}{\mathrm d(P_S P_W)}(s,w)$. 
By Tonelli's theorem, 
\[
\int_{\mathcal S} L(s,w)\,\mathrm dP_S(s)=1 \quad \text{for }P_W\text{-a.e.\ }w.
\]
For such $w$, define probability measure $P_w$ on $(\mathcal S,\mathcal F_{\mathcal S})$ by
\[
\mathrm dP_w(s):=L(s,w)\,\mathrm dP_S(s) 
\]
so that $\frac{\mathrm dP_w}{\mathrm dP_S}(s)=L(s,w)$. 
Given measurable $E\subseteq\mathcal S\times\mathcal W$, denote $E_w:=\{s\in\mathcal S:(s,w)\in E\}$. 

We apply Theorem~\ref{thm::E_gamma_converse_generalized_Orlicz} on the space $(\mathcal S,P_S)$ with Orlicz function $\varphi$, taking
\[
P=P_w,\quad Q=P_S,\quad E=E_w,
\]
to get, for $P_W$-a.e.\ $w$,
\begin{equation}
P_w(E_w)
\leq \gamma\,P_S(E_w)
+\bigl\Vert \mathds{1}_{E_w}\bigr\Vert _{\varphi}^{P_S}\,
\Bigl\Vert \bigl[L(\cdot,w)-\gamma\bigr]_+\Bigr\Vert _{\varphi^\star}^{A,P_S}.
\label{eq:fiber_bound_short}
\end{equation}

Then we integrate over $w$ by using $P_{SW}=L\cdot(P_S P_W)$ and the definition of $P_w$, and get
\[
P_{SW}(E)=\int_{\mathcal W}\int_{\mathcal S}\mathds{1}_{E_w}(s)\,L(s,w)\,\mathrm dP_S(s)\,\mathrm dP_W(w)
=\int_{\mathcal W} P_w(E_w)\,\mathrm dP_W(w).
\]
Integrating \eqref{eq:fiber_bound_short} over $w$ yields
\begin{align}
P_{SW}(E)
&\leq \gamma\int_{\mathcal W} P_S(E_w)\,\mathrm dP_W(w)
+\int_{\mathcal W}\bigl\Vert \mathds{1}_{E_w}\bigr\Vert _{\varphi}^{P_S}\,
\Bigl\Vert \bigl[L(\cdot,w)-\gamma\bigr]_+\Bigr\Vert _{\varphi^\star}^{A,P_S}\,\mathrm dP_W(w).
\label{eq:after_integrate_short}
\end{align}
By Fubini's theorem,
\[
\int_{\mathcal W} P_S(E_w)\,\mathrm dP_W(w)
=\int_{\mathcal S\times\mathcal W}\mathds{1}_E(s,w)\,\mathrm d(P_S P_W)(s,w)
= P_S P_W(E).
\]
Substitute this into \eqref{eq:after_integrate_short}.

We then apply the generalized H\"older inequality on $(\mathcal W,P_W)$ with Orlicz function $\psi$ to
\[
w\mapsto \bigl\Vert \mathds{1}_{E_w}\bigr\Vert _{\varphi}^{P_S},
\qquad
w\mapsto \Bigl\Vert \bigl[L(\cdot,w)-\gamma\bigr]_+\Bigr\Vert _{\varphi^\star}^{A,P_S},
\]
to obtain
\[
\int_{\mathcal W}\bigl\Vert \mathds{1}_{E_w}\bigr\Vert _{\varphi}^{P_S}\,
\Bigl\Vert \bigl[L(\cdot,w)-\gamma\bigr]_+\Bigr\Vert _{\varphi^\star}^{A,P_S}\,\mathrm dP_W(w)
\le
\Bigl\Vert \bigl\Vert \mathds{1}_{S\in E_W}\bigr\Vert _{\varphi}^{P_S}\Bigr\Vert _{\psi}^{P_W}\,
\Bigl\Vert \Bigl\Vert \bigl[L-\gamma\bigr]_+\Bigr\Vert _{\varphi^\star}^{A,P_S}\Bigr\Vert _{\psi^\star}^{A,P_W}.
\]
Combining the last two displays gives
\[
P_{SW}(E)
\le
\gamma P_S P_W(E)
 + 
\Bigl \Vert \bigl\Vert \mathds{1}_{S\in E_W} 
\bigr\Vert _{\varphi}^{P_S}\Bigr\Vert _{\psi}^{P_W}\,
\Bigl\Vert \Bigl\Vert \bigl[L-\gamma\bigr]_+\Bigr\Vert _{\varphi^\star}^{A,P_S}\Bigr\Vert _{\psi^\star}^{A,P_W}, 
\]
and taking infimum over $\gamma\geq 0$ concludes the proof.
\end{proof}

\section{Results on Maximal Leakage and $\alpha$-Mutual} 
\label{app::gen_bd_ML_alpha_MI}

\subsection{Proof of Corollary~\ref{cor:gamma_ML_event}}

\begin{proof}
Fix $w\in\mathcal W$. 
Apply Proposition~\ref{prop::E_gamma_converse} with
$P=P_{S|W=w}$, $Q=P_S$ and $E = E_w$. For any $\gamma\in\mathbb{R}$,
\begin{equation}
\label{eq:cond_hs}
P_{S|W=w}(E_w) \leq \gamma\,P_S(E_w)+E_\gamma(P_{S|W=w}\Vert P_S).
\end{equation}
Let $\gamma = M(w)$, we have $E_{M(w)}(P_{S|W=w}\Vert P_S) = 0$ and~\eqref{eq:cond_hs} yields $P_{S|W=w}(E_w) \leq M(w)\,P_S(E_w)$. 
Taking expectation over $W$ gives the first inequality of~\eqref{eq:ML_event_cor3} and $P_S(E_W)\leq \operatorname*{ess\,sup}_{w}P_S(E_w)$ almost surely yields the second inequality, since $\mathbf E_{P_W}[M(W)] = \exp(\mathcal L(S\to W))$. 
\end{proof}

\subsection{Results on $\alpha$-Mutual Information}

Since maximal leakage is a special $\alpha$-mutual information~\cite{verdu2015alpha} when $\alpha\to\infty$ (defined in~\eqref{eq::def_alpha_MI}), we can use Theorem~\ref{thm::E_gamma_converse_generalized_Orlicz} to recover the $\alpha$-mutual information bound~\cite[Corollary 1]{esposito2021generalization} with a simple analysis, as shown in Corollary~\ref{cor:fiber_weighted_sibson}. 
The proof is similar to that of Corollary~\ref{cor:gamma_ML_event} and is in Appendix~\ref{app::fiber_weighted_sibson}. 
Note that~\eqref{eq:fiber_weighted_intermediate} recovers $M(w)$ in Corollary~\ref{cor:gamma_ML_event} when $\alpha\to\infty$. 
\begin{corollary}
\label{cor:fiber_weighted_sibson}
Let $P_{SW}\ll P_SP_W$. Fix $\alpha>1$ and $\tau:=\frac{\alpha}{\alpha-1}$, for any $w\in\mathcal W$ and measurable $E$,
\begin{align}
P_{SW}(E)
&\leq 
\mathbf E_{P_W} \big[ M_\alpha(W)\, P_S(E_W)^{1/\tau} \big]
\label{eq:fiber_weighted_intermediate}\\
&\leq 
\big(\operatorname*{ess\,sup}_{w\sim P_W}P_S(E_w)\big)^{1/\tau}\,
\exp \big({(\alpha-1)}/{\alpha}\cdot I_\alpha(S, W)\big),  \label{eq:fiber_weighted_to_sibson}
\end{align}
where $M_\alpha(w):=
\big(\mathbf E_{P_S}\big[
\big(
{\mathrm{d} P_{S|W=w}}/{\mathrm{d} P_S}(s)
\big)^\alpha
\big]\big)^{1/\alpha}$. 
\end{corollary}

\subsection{Proof of Corollary~\ref{cor:fiber_weighted_sibson}}
\label{app::fiber_weighted_sibson}

In this section we show how to use Theorem~\ref{thm::gamma_refined_orlicz} to prove~\cite[Corollary 1]{esposito2021generalization}, with a simple analysis. 
The idea is similar to the proof of Corollary~\ref{cor:gamma_ML_event}. 

\begin{proof}
Fix $w\in\mathcal W$. 
Apply Theorem~\ref{thm::E_gamma_converse_generalized_Orlicz} with $P=P_{S|W=w}, Q=P_S, E = E_w, \gamma = 0,$ and choose the Orlicz function $\psi(t):=\frac{t^\kappa}{\kappa}$ (so that $\psi^\star(u)=\frac{u^\alpha}{\alpha}$). 
Then for this fixed $w$, Theorem~\ref{thm::E_gamma_converse_generalized_Orlicz} gives
\begin{equation}
\label{eq:cond_orlicz_step_updated}
P_{S|W=w}(E_w)
\le
\frac{1}{\psi^{-1}(1/P_S(E_w))}\,
\left\Vert 
\frac{\mathrm{d}  P_{S|W=w}}{\mathrm{d}  P_S}
\right\Vert _{\psi^\star}^{A,P_S}.
\end{equation}
We can evaluate and derive 
\begin{align*}
\frac{1}{\psi^{-1}(1/P_S(E_w))}
& =
P_S(E_w)^{1/\kappa}\,\kappa^{-1/\kappa},\\
\left\Vert 
\frac{\mathrm{d}  P_{S|W=w}}{\mathrm{d}  P_S}
\right\Vert _{\psi^\star}^{A,P_S}
& =
\kappa^{1/\kappa}
\Bigg(
\mathbf E_{P_S}\Big[
\Big(\frac{\mathrm{d}  P_{S|W=w}}{\mathrm{d}  P_S}(S)\Big)^\alpha
\Big]
\Bigg)^{1/\alpha}
=\kappa^{1/\kappa}M_\alpha(w).
\end{align*}
Plugging these two identities into~\eqref{eq:cond_orlicz_step_updated}, we obtain~\eqref{eq:fiber_weighted_intermediate}: 
\[
P_{S|W=w}(E_w)\leq M_\alpha(w)\,P_S(E_w)^{1/\kappa}.
\]
Taking expectation over $W$ yields
\[
P_{SW}(E)
=
\mathbf E_{P_W}\big[P_{S|W}(E_W)\big]
\le
\mathbf E_{P_W} \Big[M_\alpha(W)\,P_S(E_W)^{1/\kappa}\Big].
\]

To prove~\eqref{eq:fiber_weighted_to_sibson}, note that $P_S(E_W)^{1/\kappa}\le
\big(\operatorname*{ess\,sup}_{w\sim P_W}P_S(E_w)\big)^{1/\kappa}$ almost surely, hence
\[
P_{SW}(E)
\le
\Big(\operatorname*{ess\,sup}_{w\sim P_W}P_S(E_w)\Big)^{1/\kappa}\,
\mathbf E_{P_W}[M_\alpha(W)],
\]
and by observing 
\[
\mathbf E_{P_W}[M_\alpha(W)]
=
\exp \Big(\frac{\alpha-1}{\alpha}I_\alpha(S, W)\Big),
\]
we derive~\eqref{eq:fiber_weighted_to_sibson}.
\end{proof}

\subsection{Corollary~\ref{cor:gamma_ML_event} by Data Processing Inequality of $\alpha$-Mutual Information}
\label{app::recovery_ML_DPI}

Let $P_{W_\infty}$ be a minimizer of $Q_W\mapsto D_\infty(P_{SW}\Vert P_SQ_W)$, so that
\[
I_\infty (S;W)=D_\infty(P_{SW}\Vert P_SP_{W_\infty}).
\]
For any event $E$, with $p:=P_{SW}(E)$ and $q_\infty:=(P_SP_{W_\infty})(E)$, \eqref{eq:two_point_sibson} gives
\begin{equation}
I_\infty (S;W)
\geq
D_\infty(\mathrm{Ber}(p)\Vert \mathrm{Ber}(q_\infty))
=
\log\max \left\{\frac{p}{q_\infty},\frac{1-p}{1-q_\infty}\right\},
\end{equation}
and in particular
\begin{equation}
\label{eq:p_leak_qinfty}
p\leq e^{I_\infty (S;W)}\,q_\infty.
\end{equation}

Define $m(w):=\operatorname*{ess\,sup}_{s\sim P_S} \frac{\mathrm d P_{W|S=s}}{\mathrm d P_W}(w)$. 
It is standard that $I_\infty (S;W)=\mathcal L(S \to W)$ and that a minimizer
$P_{W_\infty}$ is given by
\[
\frac{\mathrm d P_{W_\infty}}{\mathrm d P_W}(w)=\frac{m(w)}{\mathbb E_{P_W}[m(W)]},
\]
and hence $P_{W_\infty}\ll P_W$.
Apply \eqref{eq:p_leak_qinfty} with $I_\infty (S;W)=\mathcal L(S \to W)$ to obtain
\[
P_{SW}(E)\leq e^{\mathcal L(S\to W)}\,(P_SP_{W_\infty})(E)
= e^{\mathcal L(S\to W)}\int P_S(E_w)\,\mathrm d P_{W_\infty}(w).
\]
Since $P_{W_\infty}\ll P_W$, the bound
$P_S(E_w)\leq \operatorname*{ess\,sup}_{w\sim P_W} P_S(E_w)$ holds $P_{W_\infty}$-a.s., so
\[
\int P_S(E_w)\,\mathrm d P_{W_\infty}(w)
\leq \operatorname*{ess\,sup}_{w\sim P_W} P_S(E_w).
\]
Combining the last two displays yields the desired result.

\section{Proof and Discussions of Corollary~\ref{cor::PAC_baye_Hellinger_div}}
\label{app::PAC_baye_Hellinger_div}

In this section, we prove Corollary~\ref{cor::PAC_baye_Hellinger_div} by using the change of measure inequality from Section~\ref{sec::change_mea_ineq} as follows: 
\begin{equation}
\label{eq::event_Hellinger_used}
    P(E)^\beta Q(E)^{1-\beta}
    +(1-P(E))^\beta(1-Q(E))^{1-\beta}
    \leq
    1+(\beta-1)\mathcal H_\beta(P\Vert Q).
\end{equation}

For any fixed $S=s$, let
\[
    a:=\sqrt{\frac{n}{2\sigma^2}},
    \qquad
    \alpha:=1-\frac{1}{\beta}=\frac{\beta-1}{\beta},
\]
and define
\[
    A(s):=1+(\beta-1)\mathcal H_\beta(P_{W|S=s}\Vert P_W).
\]
For each $y\in\mathbb R$, let
\[
    E_y(s)
    :=
    \left\{
        w:
        a\,\mathrm{gen}(s,w)>y
    \right\}.
\]
Taking $P=P_{W|S=s}$, $Q=P_W$, and $E=E_y(s)$ in~\eqref{eq::event_Hellinger_used}, we obtain
\[
    P_{W|S=s}(E_y(s))^\beta P_W(E_y(s))^{1-\beta}
    \leq A(s).
\]
Hence,
\begin{equation}
\label{eq::tail_change_measure_Hellinger}
    P_{W|S=s}(E_y(s))
    \leq
    A(s)^{1/\beta} P_W(E_y(s))^{1-1/\beta}.
\end{equation}
Using the tail integral representation, for any fixed $S=s$, we have
\begin{align}
    &\frac{
        \mathbf{E}_{P_{W|S=s}}
        \left[
            \exp\left(a\,\mathrm{gen}(s,W)\right)
        \right]
    }{A(s)^{1/\beta}}
    \nonumber\\
    &=
    \int_{-\infty}^{\infty}
    \frac{
        P_{W|S=s}\left(
            a\,\mathrm{gen}(s,W)>y
        \right)
    }{A(s)^{1/\beta}}
    e^y\,\mathrm{d}y
    \nonumber\\
    &\leq
    \int_{-\infty}^{\infty}
    P_W\left(
        a\,\mathrm{gen}(s,W)>y
    \right)^{1-1/\beta}
    e^y\,\mathrm{d}y,
    \label{eq::normalized_exp_tail_bd}
\end{align}
where the last step follows from~\eqref{eq::tail_change_measure_Hellinger}.

By Jensen's inequality, we also have
\[
    \exp\left(
        \mathbf{E}_{P_{W|S=s}}
        \left[
            a\,\mathrm{gen}(s,W)
        \right]
    \right)
    \leq
    \mathbf{E}_{P_{W|S=s}}
    \left[
        \exp\left(a\,\mathrm{gen}(s,W)\right)
    \right].
\]
Combining this with~\eqref{eq::normalized_exp_tail_bd} gives
\begin{align}
    &\frac{
        \exp\left(
            \mathbf{E}_{P_{W|S=s}}
            \left[
                a\,\mathrm{gen}(s,W)
            \right]
        \right)
    }{A(s)^{1/\beta}}
    \nonumber\\
    &\leq
    \int_{-\infty}^{\infty}
    P_W\left(
        a\,\mathrm{gen}(s,W)>y
    \right)^{1-1/\beta}
    e^y\,\mathrm{d}y.
    \label{eq::normalized_jensen_tail_bd}
\end{align}

Take expectation with respect to $S\sim P_S$. For $y\geq 0$, by the
$\sigma$-sub-Gaussian assumption,
\[
    \mathbf{E}_{P_S}
    \left[
        P_W\left(
            a\,\mathrm{gen}(S,W)>y
        \right)
    \right]
    =
    P_S P_W\left(
        a\,\mathrm{gen}(S,W)>y
    \right)
    \leq
    e^{-y^2}.
\]
Since $x\mapsto x^{1-1/\beta}$ is concave on $[0,\infty)$, we have
\[
    \mathbf{E}_{P_S}
    \left[
        P_W\left(
            a\,\mathrm{gen}(S,W)>y
        \right)^{1-1/\beta}
    \right]
    \leq
    e^{-\left(1-\frac{1}{\beta}\right)y^2},
    \qquad y\geq 0.
\]
For $y<0$, we simply use
\[
    P_W\left(
        a\,\mathrm{gen}(S,W)>y
    \right)^{1-1/\beta}
    \leq 1.
\]
Therefore, from~\eqref{eq::normalized_jensen_tail_bd},
\begin{align}
    &\mathbf{E}_{P_S}
    \left[
        \frac{
            \exp\left(
                \mathbf{E}_{P_{W|S}}
                \left[
                    a\,\mathrm{gen}(S,W)
                \right]
            \right)
        }{
            \left(1+(\beta-1)\mathcal H_\beta(P_{W|S}\Vert P_W)\right)^{1/\beta}
        }
    \right]
    \nonumber\\
    &\leq
    \int_{-\infty}^{0} e^y\,\mathrm{d}y
    +
    \int_0^\infty
    \exp\left(
        -\left(1-\frac{1}{\beta}\right)y^2+y
    \right)
    \mathrm{d}y
    \nonumber\\
    &=
    1+
    \int_0^\infty
    \exp\left(
        -\frac{\beta-1}{\beta}y^2+y
    \right)
    \mathrm{d}y.
    \label{eq::C_beta_plus_intermediate}
\end{align}
Furthermore,
\begin{align}
    1+
    \int_0^\infty
    \exp\left(
        -\frac{\beta-1}{\beta}y^2+y
    \right)
    \mathrm{d}y
    &\leq
    1+
    \int_{-\infty}^{\infty}
    \exp\left(
        -\frac{\beta-1}{\beta}y^2+y
    \right)
    \mathrm{d}y
    \nonumber\\
    &=
    1+
    \sqrt{\frac{\pi\beta}{\beta-1}}
    \exp\left(
        \frac{\beta}{4(\beta-1)}
    \right)
    \nonumber\\
    &\leq
    2\sqrt{\frac{\pi\beta}{\beta-1}}
    \exp\left(
        \frac{\beta}{4(\beta-1)}
    \right).
    \label{eq::C_beta_plus_clean_relaxation}
\end{align}
Combining~\eqref{eq::C_beta_plus_intermediate} and~\eqref{eq::C_beta_plus_clean_relaxation}, we obtain
\begin{align}
    &\mathbf{E}_{P_S}
    \left[
        \frac{
            \exp\left(
                \mathbf{E}_{P_{W|S}}
                \left[
                    a\,\mathrm{gen}(S,W)
                \right]
            \right)
        }{
            \left(1+(\beta-1)\mathcal H_\beta(P_{W|S}\Vert P_W)\right)^{1/\beta}
        }
    \right]
    \nonumber\\
    &\leq
    2\sqrt{\frac{\pi\beta}{\beta-1}}
    \exp\left(
        \frac{\beta}{4(\beta-1)}
    \right).
    \label{eq::normalized_expectation_clean_bound}
\end{align}

Applying Markov's inequality to the nonnegative random variable $\frac{
        \exp\left(
            \mathbf{E}_{P_{W|S}}
            \left[
                a\,\mathrm{gen}(S,W)
            \right]
        \right)
    }{
        \left(1+(\beta-1)\mathcal H_\beta(P_{W|S}\Vert P_W)\right)^{1/\beta}
    }$, we have that, with probability at least $1-\delta$, 
\begin{align*}
    \exp\left(
        \mathbf{E}_{P_{W|S}}
        \left[
            a\,\mathrm{gen}(S,W)
        \right]
    \right)
    \leq
    \frac{2}{\delta}
    \left(1+(\beta-1)\mathcal H_\beta(P_{W|S}\Vert P_W)\right)^{1/\beta}
    \sqrt{\frac{\pi\beta}{\beta-1}}
    \exp\left(
        \frac{\beta}{4(\beta-1)}
    \right).
\end{align*}
Taking logarithm on both sides and recalling that
$a=\sqrt{n/(2\sigma^2)}$, we obtain
\begin{align*}
    \mathbf{E}_{P_{W|S}}[\mathrm{gen}(S,W)]
    \leq
    \sqrt{\frac{2\sigma^2}{n}}
    \left(
        \log 2
        + \log \sqrt{\frac{\pi\beta}{\beta-1}}
        + \frac{\beta}{4(\beta-1)}
        + \log \left(
            \frac{
                \left(1+(\beta-1)\mathcal H_\beta(P_{W|S}\Vert P_W)\right)^{1/\beta}
            }{\delta}
        \right)
    \right).
\end{align*}

\section{Generalization Error Bounds via Conditional Mutual Information}
\label{sec::CMI}

Conditional Mutual Information (CMI) framework was proposed by~\cite{steinke2020reasoning} for generalization error analysis (also see recent works~\cite{wang2024generalization, sefidgaran2025tighter}). 
Unlike mutual information, which can easily be infinite even in settings where generalization is easy to prove, conditional mutual information is always finite.

In this section, we employ a strategy that is similar to~\cite{hellstrom2020generalization}, and derive CMI bounds from our change of measure inequalities.
We make the same assumption as~\cite{hellstrom2020generalization} that the loss function $\ell(\cdot, \cdot)$ is bounded on $[a,b]$ instead of in this subsection.

Let $\tilde{Z}=(Z_1,\dots,Z_{2n})$ be the super-sample containing $2n$ i.i.d. tranining samples generated from $P_Z$, $S=(S_1,\dots,S_n)\in\{0,1\}^n$ be a random selection vector that is independent of $\tilde{Z}$, and $W$ be the output of a learning algorithm that may depend on $\tilde{Z}$. 
Use $Z(S)$ to denote the subset of $\tilde{Z}$ obtained from $S$ by taking $Z_i(S_i) = \tilde{Z}_{i + S_i n}$ for $i=1,\ldots,n$. 
We then have $P_{W|Z(S)}$ as a stochastic learning algorithm, where $W$ and $(\tilde{Z}, S)$ are conditionally independent given $Z(S)$.

Define 
\begin{equation}
    \widehat{\mathrm{gen}}(W,\tilde{Z},S) := \frac{1}{n} \sum^n_{i=1}\big(
    \ell(W, Z_i(\bar{S}_i) ) - \ell(W, Z_i(S_i) )
    \big),  \label{eq::complement_gen_error}
\end{equation}
where $\bar{S}$ is a vector whose entries are modulo-$2$ complements of the entries of $S$. As a result, $Z(\bar{S})$ contains all the elements of $\tilde{Z}$ that are not included in $Z(S)$. 
One can observe that the quantities like~\eqref{eq::complement_gen_error} are what being empirically calculated in the generalization performance assessment of an algorithm in practice. 
\cite[Theorem 3]{hellstrom2020generalization}, provides a way to convert the bound based on~\eqref{eq::complement_gen_error} to a generalization bound: under $P_{W\tilde{Z}S}$, if with probability at least $1 - \delta$ we
\begin{equation*}
    \widehat{\mathrm{gen}}(W,\tilde{Z},S) \leq \epsilon(\delta)
\end{equation*}
holds with probability at least $1 - \delta$,
then also with probability at least $1-\delta$ we have
\begin{equation*}
    \big|\mathrm{gen}(W,Z(S)) \big| \leq \epsilon\left( \delta/2 \right) + \sqrt{((b-a)^2/(2n))\cdot  \log(4/\delta}).
\end{equation*}
We then present bounds for $\widehat{\mathrm{gen}}(W,\tilde{Z},S)$ using Theorem~\ref{thm::E_gamma_converse_generalized_Orlicz}. Note $P:=P_{W\tilde{Z}S}$, $Q:=P_{W|\tilde{Z}}P_{\tilde{Z}S}$.
\begin{theorem}
\label{thm::CMI_general}
Let $\psi$ be an Orlicz function and $\psi^\star$ its conjugate.  For any $\gamma\in\mathbb{R}$ and $\eta>0$, 
\[
\mathbf{P}_{P_{W\tilde{Z}S}} \big(\big|\widehat{\mathrm{gen}}(W,\tilde{Z},S)\big|\geq \eta\big)
\leq 
\frac{1}{\psi^{-1} \big(\frac{1}{2}\exp \big(\frac{n\eta^2}{2(b-a)^2}\big)\big)}
\Big\Vert \Big[\frac{\mathrm{d}P}{\mathrm{d}Q}-\gamma\Big]_+ \Big\Vert_{\psi^\star}^{A,Q}
+
2[\gamma]_+\exp \Big(\frac{-n\eta^2}{2(b-a)^2}\Big).
\]
\end{theorem}
The proof of Theorem~\ref{thm::CMI_general} is in Appendix~\ref{app::CMI_general}, and it can be specialized by Proposition~\ref{prop::E_gamma_converse} as follows.

\begin{corollary}
\label{cor::CMI_Egamma_tail_equiv}
Fix any $\gamma\in\mathbb{R}$. Then for every $\eta>0$,
\[
\mathbf{P}_{P_{W\tilde{Z}S}} \left(\big|\widehat{\mathrm{gen}}(W,\tilde{Z},S)\big|\geq \eta\right)
\leq 
E_\gamma \big(P \big\Vert Q \big)
+ 
2[\gamma]_+ \exp \big(-{n\eta^2}/(2(b-a)^2)\big).
\]
\end{corollary}

Similar to the use of Theorem~\ref{thm::E_gamma_converse_generalized_Orlicz}, we can also employ other change of measure inequalities in Section~\ref{sec::change_mea_ineq} to derive CMI bounds in terms of other measures. 
They are omitted due to limited space.

\subsection{Proof of Theorem~\ref{thm::CMI_general}}
\label{app::CMI_general}

\begin{proof}
We apply  Theorem~\ref{thm::E_gamma_converse_generalized_Orlicz} with
\[
P := P_{W\tilde{Z}S},\quad Q := P_{W|\tilde{Z}}P_{\tilde{Z}S},\quad
E := \{(W,\tilde{Z},S): \big|\widehat{\mathrm{gen}}(W,\tilde{Z},S) \big|\geq \eta \}.
\]

Similar to~\cite{hellstrom2020generalization}, define the \emph{fiber} of $E$
with respect to $(W,\tilde{Z})$ by
\[
E_{W\tilde Z}:=\{s\in\mathcal S:(W,\tilde Z,s)\in E\}.
\]

Since under $Q=P_{W|\tilde Z}P_{\tilde Z S}$ we have $W\perp S| \tilde Z$, it follows that
\begin{align}
Q(E)
&=\mathbb E_Q\big[\mathds 1\{(W,\tilde Z,S)\in E\}\big]\nonumber\\
&=\mathbb E_{P_{\tilde Z}}\Big[\mathbb E_{P_{W|\tilde Z}}\big[ P_{S|\tilde Z}(E_{W\tilde Z}| \tilde{Z})\big]\Big]
=\mathbb E_{P_{W\tilde Z}}\big[ P_{S|\tilde Z}(E_{W\tilde Z}| \tilde{Z})\big].
\label{eq::QE_fiber}
\end{align}

Since $S$ is independent of $\tilde Z$, under $S\sim P_{S|\tilde Z=\tilde z}$ the coordinates $S_1,\ldots,S_n$ are i.i.d.\ Bernoulli$(1/2)$.
For each fixed $(w,\tilde{Z})$, define
\[
\Delta_i(w,\tilde{Z}):=\ell(w,\tilde z_{i+n})-\ell(w,\tilde z_i),\qquad i=1,\dots,n.
\]
Then
\[
\widehat{\mathrm{gen}}(w,\tilde z,S)
=
\frac{1}{n}\sum_{i=1}^n (1-2S_i)\Delta_i(w,\tilde{Z}).
\]
Hence
$\mathbb E_{S\sim P_{S|\tilde Z=\tilde z}}\big[\widehat{\mathrm{gen}}(w,\tilde z,S)\big]=0$, and each summand belongs to
$\left[-\frac{b-a}{n},\frac{b-a}{n}\right]$ since $|\Delta_i(w,\tilde{Z})|\le b-a$.
By Hoeffding's lemma, $\widehat{\mathrm{gen}}(w,\tilde z,S)$ is $(b-a)/\sqrt n$-sub-Gaussian under
$S\sim P_{S|\tilde Z=\tilde z}$. Therefore, for all $(w,\tilde{Z})$,
\[
P_{S|\tilde Z=\tilde z}\big(\,|\widehat{\mathrm{gen}}(w,\tilde z,S)|\geq \eta \,\big)
\leq 2\exp \left(-\frac{n\eta^2}{2(b-a)^2}\right).
\]
Combining this uniform bound with~\eqref{eq::QE_fiber} yields
\[
Q(E)\leq 2\exp \left( -\frac{n\eta^2}{2(b-a)^2} \right).
\]

\medskip

We now apply Theorem~\ref{thm::E_gamma_converse_generalized_Orlicz} with the above $P,Q,E$, and derive
\begin{align*}
P_{W\tilde{Z}S}(E)
&\leq \gamma\,Q(E)
 + \frac{1}{\psi^{-1} \big(1/Q(E)\big)}
   \left\Vert \left[\frac{\mathrm{d}P}{\mathrm{d}Q}-\gamma\right]_+ \right\Vert_{\psi^\star}^{A,Q} \\
&\textcolor{black}{\leq [\gamma]_+\,Q(E)}
 + \frac{1}{\psi^{-1} \big(1/Q(E)\big)}
   \left\Vert \left[\frac{\mathrm{d}P}{\mathrm{d}Q}-\gamma\right]_+ \right\Vert_{\psi^\star}^{A,Q} \\
&\leq \textcolor{black}{2[\gamma]_+}\exp \left(-\frac{n\eta^2}{2(b-a)^2}\right)
 + \frac{1}{\psi^{-1} \left(\frac{1}{2}\exp \left(\frac{n\eta^2}{2(b-a)^2}\right)\right)}
   \left\Vert \left[\frac{\mathrm{d}P}{\mathrm{d}Q}-\gamma\right]_+ \right\Vert_{\psi^\star}^{A,Q}.
\end{align*}
This proves the desired result.
\end{proof}

\section{Proof of Lemma~\ref{lem:approxmaxinfo_Egamma_equiv}}
\label{app::approxmaxinfo_Egamma_equiv}

\begin{proof}
Recall that
\[
E_\gamma(P\Vert Q)=\sup_{E}\bigl(P(E)-\gamma Q(E)\bigr).
\]

\noindent($\Rightarrow$)
Assume $E_\gamma(P\Vert Q)\leq \tau$. Then for every measurable event $E$,
\[
P(E)-\gamma Q(E)\leq \tau. 
\]
Fix any $E$ such that $P(E)>\tau$. If $Q(E)=0$, then the above inequality would give
$P(E)\leq \tau$, a contradiction. Hence $Q(E)>0$ and we may divide to obtain
\[
\frac{P(E)-\tau}{Q(E)}\leq \gamma.
\]
Taking the supremum over all $E$ with $P(E)>\tau$ yields
\[
\sup_{E: P(E)>\tau} \frac{P(E)-\tau}{Q(E)}\leq \gamma,
\]
which is equivalent to $I_\infty^\tau(P\Vert Q)\leq \log\gamma$.

\noindent($\Leftarrow$)
Assume $I_\infty^\tau(P\Vert Q)\leq \log\gamma$, i.e.,
\[
\sup_{E: P(E)>\tau}\frac{P(E)-\tau}{Q(E)}\leq \gamma.
\]
Fix any measurable event $E$. If $P(E)\leq \tau$, then trivially
$P(E)-\gamma Q(E)\leq P(E)\leq \tau$.
If $P(E)>\tau$, then (by the definition of the supremum) we have $Q(E)>0$ and
\[
\frac{P(E)-\tau}{Q(E)}\leq \gamma,
\]
hence $P(E)\leq \gamma Q(E)+\tau$, i.e.\ $P(E)-\gamma Q(E)\leq \tau$.
Since this holds for all $E$, taking the supremum over $E$ gives
$E_\gamma(P\Vert Q)\leq \tau$.
\end{proof}

\section{Generalization Error Bounds and Differential Privacy}
\label{app::gen_DP}

We compare our generalization error bounds in terms of $\chi^2(P\Vert Q)$, as displayed in Theorem~\ref{thm::gen_bd_f_div}, with the best-known generalization bound in terms of maximal leakage~\cite{esposito2021generalization}, as follows.

\begin{figure}[htpb]
    \centering
    \includegraphics[scale=0.4]{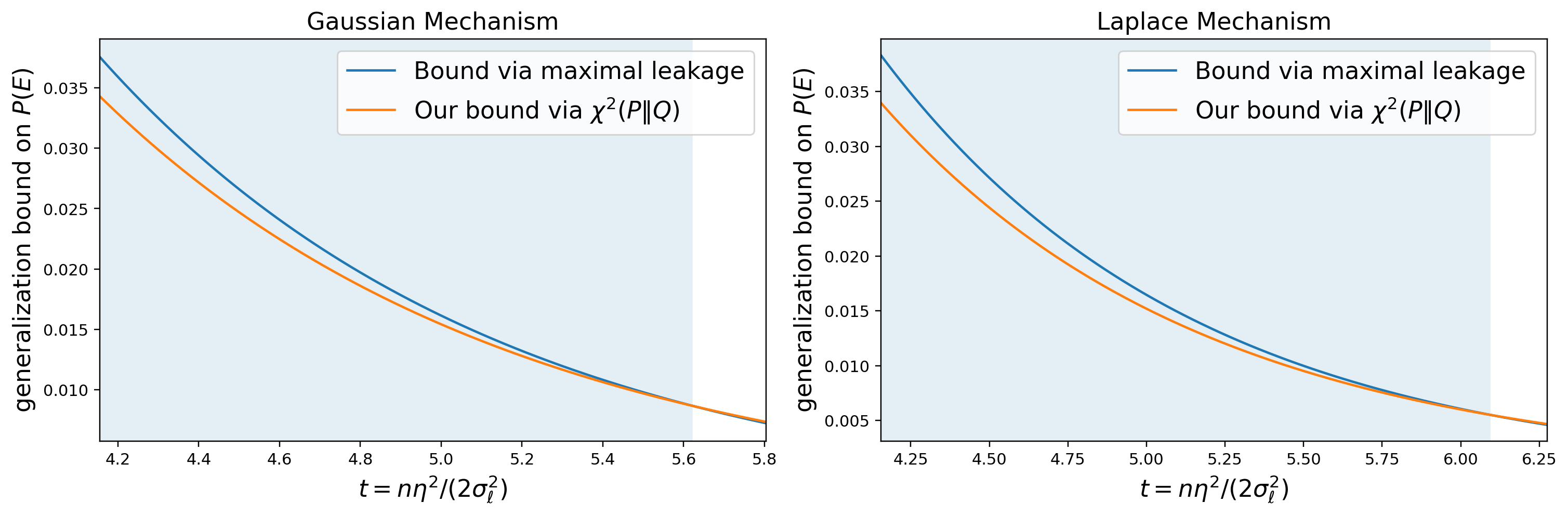}
    \caption{Comparison of generalization bounds of differential privacy algorithms, with binary skewed source distribution $\Pr(S = \Delta)=10^{-3}$ and output $W=S+N$. 
    Let $P=P_{SW}$, $Q=P_S P_W$, $Q(E)\leq q_t = 2e^{-t}$ and $t=n\eta^2/(2\sigma_\ell^2)$. 
    For both Gaussian ($N\sim\mathcal N(0,\tau^2)$, $\Delta/(2\tau)=0.25$) and Laplace ($N\sim\mathrm{Lap}(b)$, $\Delta/(2b)=0.25$) noises, bounds via $\chi^2(P\Vert Q)$ are \textbf{tighter} than the bounds via maximal leakage~\cite{esposito2021generalization}. } 
    \label{fig::DP_gen}
\end{figure}

\section{Average Generalization Error Bound}
\label{app:proof_gen_bd_MI}

By~\cite{xu2017information}, an average generalization bound in terms of $I(S;W)$ has been proved:
\begin{equation}
    \label{eq::xu_gen_bd_MI}
    \mathbf{E}\left[
    \mathrm{gen}(S, W)
    \right] 
    \leq \sqrt{(2\sigma^2/n) I(S; W)},
\end{equation}

In this section, we show that our change of measure inequalities can also be used to derive average generalization bounds, by deriving, as an example, a bound in terms of $I(S;W)$ that is close to~\eqref{eq::xu_gen_bd_MI} (up to a multiplicative constant) by using our Proposition~\ref{prop::E_gamma_converse}, a simple change of measure inequality in terms of $E-\gamma$-divergence. 
We also show that our bound  is strictly tighter than the one recovered by~\cite[Corollary~1]{chu2023unified}.

If, instead of using Proposition~\ref{prop::E_gamma_converse}, we use our other results (e.g., Theorem~\ref{thm::E_gamma_converse_generalized_Orlicz}), then average generalization bounds in terms of other information measures can also be derived.

\subsection{An Average Generalization Bound via $I(S;W)$}

\begin{corollary}
    \label{cor::gen_bd_MI}
    For a learning algorithm $P_{W|S}$ with a $\sigma$-sub-Gaussian loss function, 
    \begin{equation*}
        \mathbf{E}\left[
        |\mathrm{gen}(S, W)|
        \right] \leq 
        \big(2\sigma^2/n \big)
        \left(2\sqrt{I(S;W)+2/e} + \sqrt{\pi}\right).
    \end{equation*}
\end{corollary}

\begin{proof}
    Take $P = P_{SW}$ and $Q = P_SP_W$, and $E = \left\{ W, S\,:\, \left| \mathrm{gen}(S, W) \right|\geq \epsilon \right\}$, and we consider the $\sigma$-sub-Gaussianity of the loss function. 
    By Hoeffding's inequality we know 
    \begin{equation*}
        Q(E) \leq 2\exp\left(
        -\frac{n\epsilon^2}{2\sigma^2}
        \right). 
    \end{equation*}

    Apply Proposition~\ref{prop::E_gamma_converse}, for $\gamma > 0$, we have 
    \begin{equation*}
        \mathbf{P}\left(
        \left| \mathrm{gen}(S, W) \right|\geq \epsilon
        \right) 
        \leq
        2\gamma \exp\left(
        -\frac{n\epsilon^2}{2\sigma^2}
        \right) 
        + E_\gamma (P_{SW}\Vert P_SP_W). 
    \end{equation*}

    By~\cite[Theorem 29]{sason2016f} we know 
    \begin{equation*}
        E_\gamma (P_{SW}\Vert P_SP_W)\leq \frac{I(S;W)+ 2/e}{\ln \gamma},
    \end{equation*}
    and therefore
    \begin{equation}
        \mathbf{P}\left(
        \left| \mathrm{gen}(S, W) \right|\geq \epsilon
        \right) 
        \leq
        2\gamma \exp\left(
        -\frac{n\epsilon^2}{2\sigma^2}
        \right) 
        + \frac{I(S;W)+ 2/e}{\ln \gamma}. 
        \label{eq::MI_genbd_temp1}
    \end{equation}

    We now need to convert the tail bound into an expectation bound. 
    
    Consider for nonnegative $X$ we know 
    \begin{equation*}
        \mathbf{E}[X] = \int_0^\infty \mathbf{P}(X\geq \epsilon)\mathrm{d}\epsilon.
    \end{equation*}

    Take $X := \left| \mathrm{gen}(S, W) \right|$. 
    To make~\eqref{eq::MI_genbd_temp1} integrable, we choose $\gamma$ depending on $\epsilon$: 
    \begin{equation*}
        \gamma_\epsilon = \exp\left(
        \frac{n\epsilon^2}{4\sigma^2}
        \right), 
    \end{equation*}
    and then we get from~\eqref{eq::MI_genbd_temp1} that
    \begin{equation}
        \mathbf{P}\left(
        \left| \mathrm{gen}(S, W) \right|\geq \epsilon
        \right) 
        \leq
        2 \exp\left(
        - \frac{n\epsilon^2}{4\sigma^2}
        \right) 
        +
        \frac{4\sigma^2}{n\epsilon^2}\left(I(S;W)+ \frac{2}{e}
         \right).
        \label{eq::MI_genbd_temp2}
    \end{equation}

    We then integrate~\eqref{eq::MI_genbd_temp2} as follows, 
    \begin{align}
        & \mathbf{E}\left[
        |\mathrm{gen}(S, W)|
        \right] \nonumber \\
        & = \int_0^{\epsilon_0} \mathbf{P} \big(\left| \mathrm{gen}(S, W) \right|\geq \epsilon \big) \mathrm{d}\epsilon + \int^\infty_{\epsilon_0} \mathbf{P}\big( \left| \mathrm{gen}(S, W) \right|\geq \epsilon \big)\mathrm{d}\epsilon \nonumber \\
        &\leq \epsilon_0 + 
        2\int^\infty_0 \exp\left(
        - \frac{n\epsilon^2}{4\sigma^2}
        \right) \mathrm{d}\epsilon + \frac{4\sigma^2}{n}\left(I(S;W)+ \frac{2}{e}
         \right) \int^\infty_{\epsilon_0}\frac{1}{\epsilon^2} \mathrm{d}\epsilon. \label{eq::MI_genbd_temp4}
    \end{align}
    We can then compute 
    \begin{align*}
        2\int^\infty_0 \exp\left(
        - \frac{n\epsilon^2}{4\sigma^2}
        \right) \mathrm{d}\epsilon & = \frac{2\sigma\sqrt{\pi}}{\sqrt{n}},\\
        \int^\infty_{\epsilon_0}\frac{1}{\epsilon^2} \mathrm{d}\epsilon &= \frac{1}{\epsilon_0},
    \end{align*}
    and therefore 
    \begin{equation}
        \mathbf{E}\left[
        |\mathrm{gen}(S, W)|
        \right] \leq \epsilon_0 + \frac{2\sigma\sqrt{\pi}}{\sqrt{n}} + \frac{4\sigma^2}{n\epsilon_0} \left(I(S;W)+ \frac{2}{e}
         \right),
          \label{eq::MI_genbd_temp3}
    \end{equation}
    on which we optimize the right-hand side over $\epsilon_0$ by balancing the first and third terms and find 
    \begin{equation*}
        \epsilon_0^\star = 2\sigma\sqrt{\frac{(I(S;W)+ 2/e}{n}}.
    \end{equation*}
    Plug above into~\eqref{eq::MI_genbd_temp3} we get 
    \begin{align*}
        \mathbf{E}\left[
        |\mathrm{gen}(S, W)|
        \right] & \leq 
        4 \sigma\sqrt{\frac{I(S;W)+ 2/e}{n}} + \frac{2\sigma\sqrt{\pi}}{\sqrt{n}}  \\
        & = \frac{2\sigma}{\sqrt{n}}\left(2\sqrt{I(S;W)+\frac{2}{e}} + \sqrt{\pi}\right).
    \end{align*}

\end{proof}

\subsection{Comparison and Discussion}

Note in the proof of Corollary~\ref{cor::gen_bd_MI}, for step~\eqref{eq::MI_genbd_temp4} we calculated $\int^\infty_0 \exp\left( - \frac{n\epsilon^2}{4\sigma^2} \right) \mathrm{d}\epsilon$ instead of calculating $\int^\infty_{\epsilon_0} \exp\left( - \frac{n\epsilon^2}{4\sigma^2} \right) \mathrm{d}\epsilon$; if we stick to the former, and keep the rest of the proof the same, result that is tighter but in a more complicated form can be derived: 
\begin{equation*}
    \mathbf{E}\big[\,|\mathrm{gen}(S, W)|\,\big]
     \leq 
    \frac{2\sigma}{\sqrt{n}}
    \left(
    2 t^\star\big(1 - e^{-(t^\star)^2}\big)
    + \sqrt{\pi}\,\mathrm{erfc}(t^\star)
    \right),
\end{equation*}
where $t^\star$ is the unique solution of 
\begin{equation*}
    (t^\star)^2\Big(1 - 2 e^{-(t^\star)^2}\Big) =  I(S;W) + \frac{2}{e}. 
\end{equation*}

Comparing to what can be recovered by~\cite{chu2023unified}, Corollary~\ref{cor::gen_bd_MI} is strictly tighter, and both of which are  weaker than the original from by~\cite{xu2017information} up to some multiplicative constant factor, but still in the same order. 
Now we provide a detailed comparison. 

In~\cite[Corollary 1]{chu2023unified} the following bound is recovered: 
\begin{equation*}
    \mathbf{E}\left[
    |\mathrm{gen}(S, W)|
    \right] \leq \frac{2\sigma}{\sqrt{n}}\sqrt{6\big(I(S;W)+4\big)}.
\end{equation*}

Since both share the factor $\sigma/\sqrt{n}$, it suffices to compare two right-hand sides in reduced forms
\[
4\sqrt{I(S;W)+\frac{2}{e}}+2\sqrt{\pi}
\quad\text{and}\quad
\sqrt{24\big(I(S;W)+4\big)}.
\]
Define their difference
\[
\Delta
:=\sqrt{24\big(I(S;W)+4\big)}-4\sqrt{I(S;W)+\frac{2}{e}}-2\sqrt{\pi},
\]
by taking derivative  with respect to $I(S;W)$ and finding the global minimizer, we find at the global minimizer
\[
\Delta 
=2\Big(\sqrt{8-\frac{4}{e}}-\sqrt{\pi}\Big)\approx 1.5653>0,
\]
and hence our Corollary~\ref{cor::gen_bd_MI} is strictly tighter and the gap is uniformly around $1.5653 \frac{\sigma}{\sqrt{n}}$. 
See Figure~\ref{fig::Comparison_MI} for comparison. 
\begin{figure}[htpb]
    \centering
    \includegraphics[scale=0.6]{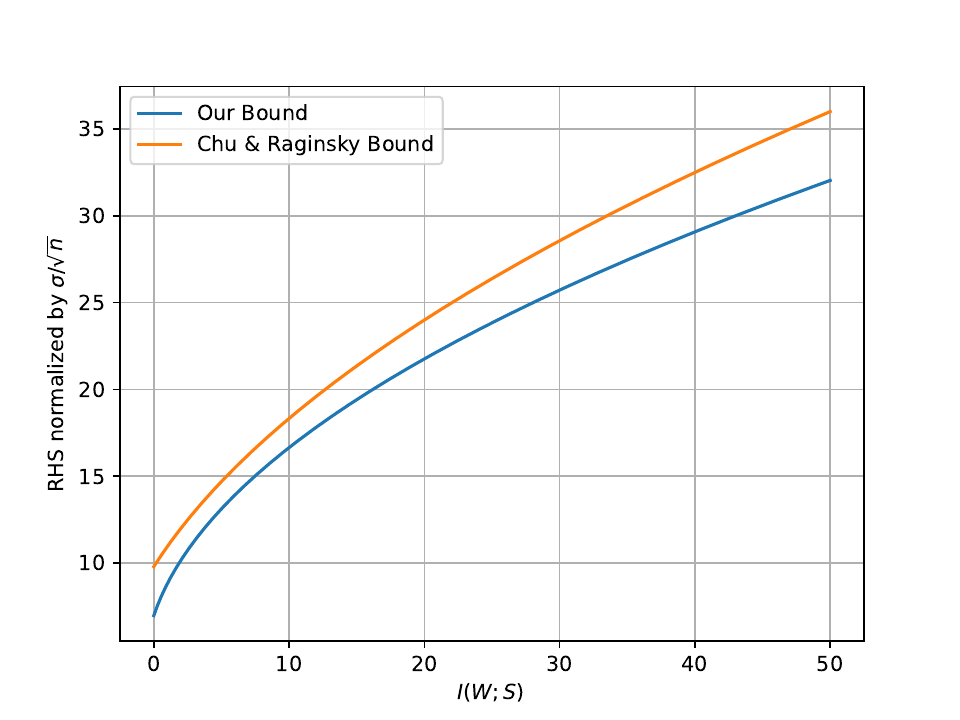}
    \caption{Comparison between Corollary~\ref{cor::gen_bd_MI} and~\cite[Corollary 1]{chu2023unified}. } 
    \label{fig::Comparison_MI}
\end{figure}

\section{Details of Table~\ref{tab:com_event_specialized_with_dpi_partial}}
\label{app::com_event_specialized_with_dpi}

The proofs in this section are based on~\eqref{eq::DPI_f_div_detail}.
We also discuss some details for the comparison between our bounds and the results in Section~\ref{subsec::comparison}.

\subsection{Reverse Pearson $f(t)=t^{-1}-1$}
\label{subsec::dpi_reverse_pearson}

Consider the reverse Pearson divergence is an $f$-divergence with $f(t)=t^{-1}-1$, which satisfies
\[
D_f(P\Vert Q)=\int\Big(\Big(\frac{dP}{dQ}\Big)^{-1}-1\Big)\,dQ
=\chi^2(Q\Vert P).
\]
Let $r:=\chi^2(Q\Vert P)$, we have
\[
r
\ge
q\Big(\frac{q}{p}-1\Big)+(1-q)\Big(\frac{1-q}{1-p}-1\Big)
=
\frac{q^2}{p}+\frac{(1-q)^2}{1-p}-1.
\]
Equivalently,
\[
\frac{q^2}{p}+\frac{(1-q)^2}{1-p}\leq 1+r.
\]
Multiplying by $p(1-p)$ gives
\[
q^2+(1-2q)p\leq (1+r)p-(1+r)p^2,
\]
i.e.
\[
(1+r)p^2-(r+2q)p+q^2\leq 0.
\]
Thus $p$ lies between the two roots of this quadratic, and in particular
\begin{equation}
p=P(E)\le
\frac{r+2q+\sqrt{r^2+4r\,q(1-q)}}{2(1+r)},
\qquad r=\chi^2(Q\Vert P).
\label{eq::reverse_pearson_event}
\end{equation}

\subsection{Reverse KL $f(t)=-\log t$}
\label{subsec::dpi_reverse_kl}

Consider the reverse KL divergence is an $f$-divergence with $f(t)=-\log t$, so that
\[
D_f(P\Vert Q)=D(Q\Vert P).
\]
We have
\begin{equation}
D(Q\Vert P)\ge
q\log\frac{q}{p}+(1-q)\log\frac{1-q}{1-p}
=: \mathrm{kl}(q,p).
\label{eq::reverse_kl_binary}
\end{equation}
For fixed $q\in(0,1)$, the map $p\mapsto \mathrm{kl}(q,p)$ is strictly increasing on $[q,1)$.
Hence, defining $p_+(q,d)$ as the unique solution in $[q,1)$ of $\mathrm{kl}(q,p)=d$,
\eqref{eq::reverse_kl_binary} implies the sharp inversion
\begin{equation}
p=P(E)\leq p_+\big(q,\,D(Q\Vert P)\big),
\qquad
\text{where }p_+(q,d)\ \text{solves }\mathrm{kl}(q,p)=d.
\label{eq::reverse_kl_event_inverse}
\end{equation}
If one prefers an explicit (but one-sided) bound, dropping the first term in \eqref{eq::reverse_kl_binary} yields
\begin{equation}
p\leq 1-(1-q)\exp \Big(-\frac{D(Q\Vert P)}{1-q}\Big).
\label{eq::reverse_kl_event_explicit}
\end{equation}

\subsection{Vincze-Le Cam divergence and the optimized quadratic bound}
\label{subsec::dpi_vincze_lecam}

Consider the Vincze-Le Cam divergence, which is an $f$-divergence with $f(t)=(2-2t)/(t+1)$. 
Similar to above, we denote $p:=P(E)$ and $q:=Q(E)$ and calculate by the data-processing inequality: 
\begin{align}
\mathrm{VC}(P\Vert Q)
& \geq 
D_f(\mathrm{Ber}(p)\Vert \mathrm{Ber}(q))  \nonumber\\
&=
q\,f \Big(\frac{p}{q}\Big)+(1-q)\,f \Big(\frac{1-p}{1-q}\Big)  \nonumber\\
&=
-2
+\frac{4q^2}{p+q}
+\frac{4(1-q)^2}{2-p-q}, \nonumber \\
& =
\frac{2(p-q)^2}{(p+q)(2-p-q)}.
\label{eq::VC_fraction_form}
\end{align}

Let $V:=\mathrm{VC}(P\Vert Q)\ge 0$. Then \eqref{eq::VC_fraction_form} gives
\[
V(p+q)(2-p-q)\geq  2(p-q)^2.
\]
Expanding and rearranging yields
\begin{equation}
(V+2)p^2-2(Vq+2q)p+Vq^2+2q^2-2Vq \leq 0.
\label{eq::VC_quadratic_in_p}
\end{equation}
Thus $p$ lies between the two roots of \eqref{eq::VC_quadratic_in_p}, and in particular
\begin{equation}
p \leq
\frac{V(1-q)+2q+\sqrt{V\bigl(V+8q(1-q)\bigr)}}{V+2}.
\label{eq::VC_event_root}
\end{equation}

We now compare \eqref{eq::VC_event_root} with  the following result from~\cite{picard2022change}. 

For $c>0$, we first define the bound
\begin{equation}
p\leq 
2(1+c)-q
-\frac{4\Big(q\sqrt{c}+(1-q)\sqrt{1+c}\Big)^2}{V+2}
=:A(c).
\label{eq::VC_event_family_Ac}
\end{equation}

We let $r:=\sqrt{c/(1+c)}$ and hence  \eqref{eq::VC_event_family_Ac} becomes 
\begin{equation}
A(c)=h(r)
:=
\frac{2-\frac{4}{V+2}\,\bigl(1-q+qr\bigr)^2}{1-r^2}-q,
\qquad r\in(0,1).
\label{eq::h_of_r}
\end{equation}
By differentiating \eqref{eq::h_of_r} one can find 
\begin{equation}
r^\star
=
\frac{V+4 q(1-q) -\sqrt{V\bigl(V+8 q(1-q) \bigr)}}{4 q(1-q) }.
\label{eq::r_star}
\end{equation}
and one can verify $r^\star$ is the unique global minimizer of $h$ on $(0,1)$.
Let $c^\star:=\frac{(r^\star)^2}{1-(r^\star)^2}$ be the corresponding optimizer in \eqref{eq::VC_event_family_Ac}.
Then for all $c>0$,
\begin{equation}
A(c)\geq  A(c^\star)=h(r^\star).
\label{eq::Ac_ge_opt}
\end{equation}
Finally, evaluating $h(r^\star)$ using the defining quadratic relation of $r^\star$ yields
\begin{equation}
\min_{c>0}A(c)
=
h(r^\star)
=
\frac{V(1-q)+2q+\sqrt{V\bigl(V+8q(1-q)\bigr)}}{V+2}.
\label{eq::Ac_min_equals_B}
\end{equation}
Hence we know that \eqref{eq::VC_event_root} is exactly the best bound obtainable from the family \eqref{eq::VC_event_family_Ac} by optimizing over $c>0$.

\section{Comparison Table}
\label{app::table}

In this section, we provide a table that explicitly compare our results with the results by~\cite{picard2022change}, both applying on the indicator channel $\mathds{1}_E$. 
See Table~\ref{tab:com_event_specialized_with_dpi_full}.

\begin{table}[!htbp]
\small
\caption{
Our change of measure inequalities in terms of typical $f$-divergences via DPI, all of which are never worse, and usually tighter, than best-known inequalities in literature~\cite{picard2022change, esposito2021generalization, ohnishi2021novel}. For each $f$-divergence, the first line is our result and the second line is the corresponding bound by~\cite{picard2022change} (PWG), except $E_\gamma$ which is novel. 
Let $h_2(\cdot)$ be binary entropy function and $c>0, s\in\mathbb{R}, q_\beta:=\beta/(\beta-1)$. We sometimes omit $(P\Vert Q), (P;Q)$ after divergences. }

\renewcommand{\arraystretch}{1.65}
\centering
\begin{tabular}{p{2.55cm} p{3.2cm} p{8.2cm}}
\hline\hline 
$f$-div
& $f(t)$
& Change of Measure Inequalities
\\
\hline\hline

\parbox[t]{2.35cm}{$E_\gamma$-div, $\gamma\geq 1$}
& $\displaystyle [t-\gamma]_+$
& $\displaystyle p\leq \gamma\,q+E_\gamma(P\Vert Q)$
\\

KL
& $t\log t$
& $\displaystyle p\leq \big(D_{\mathrm{KL}}(P\Vert Q)+\log  \big(1+q(e^c-1)\big)\big)\big/ c$ 
\\ 

$\chi^2$-div
& $t^2-1$
& $\displaystyle p\leq q+\sqrt{q\bigl(1-q\bigr)\chi^2}$
\\ 

\parbox[t]{2.35cm}{Power-$\beta$, $\beta>1$}
& $\displaystyle (t^\beta-1) / (\beta-1)$
& $\displaystyle \begin{aligned}[t]
&p^\beta q^{1-\beta} +(1-p)^\beta(1-q)^{1-\beta} \leq 1+(\beta-1)\mathcal H_\beta
\end{aligned}$
\\
&
&
$\displaystyle \begin{aligned}[t]
&\text{PWG: } p\leq s+\bigl(1+(\beta-1)\mathcal H_\beta\bigr)^{1/\beta}\\
&\qquad \qquad \times \bigl(q(1-s)_+^{q_\beta}+(1-q)(-s)_+^{q_\beta}\bigr)^{1/q_\beta}, 
\end{aligned}$
\\

Squared Hellinger
& $(1-\sqrt{t})^2$
& $\displaystyle 2\Bigl(1-\sqrt{pq}-\sqrt{(1-p)(1-q)}\Bigr)\leq H^2$
\\
&
&
$\displaystyle \begin{aligned}[t]
&\text{PWG: } p\leq 1+c-\bigl(1-H^2/2\bigr)^2 \Bigl(q/c+(1-q)/(1+c)\Bigr)^{-1}
\end{aligned}$
\\

Reverse $\chi^2$-div
& $\displaystyle 1/t - 1$
& $\displaystyle {(p-q)^2}/(p(1-p))\leq \chi^2(Q\Vert P)$
\\
&
&
\text{PWG: } $\displaystyle p\leq 1+c-\frac{\Bigl(q\sqrt{c}+(1-q)\sqrt{1+c}\Bigr)^2}{1+\chi^2(Q\Vert P)}$, $c>0$
\\

Reverse-KL
& $-\log t$
& $\displaystyle \begin{aligned}[t]
&q\log(\frac{q}{p}) +(1-q)\log(\frac{1-q}{1-p}) \leq D_{\mathrm{KL}}(Q\Vert P)
\end{aligned}$
\\ 
&
&
$\displaystyle \begin{aligned}[t]
&\text{PWG: } p\leq 1+c-\exp\Bigl(q\log c +(1-q)\log(1+c)-D_{\mathrm{KL}}(Q\Vert P)\Bigr) 
\end{aligned}$
\\

Jensen-Shannon
& $\displaystyle t\log  \frac{2t}{1+t}+\log  \frac{2}{1+t}$
& $\displaystyle 2h_2  \left((p+q)/2\right)-h_2  \bigl(p\bigr)-h_2  \bigl(q\bigr)\leq \mathrm{JS}(P\Vert Q)$
\\
&
&
$\displaystyle \begin{aligned}[t]
&\text{PWG: }  p\leq 1+c+\mathrm{JS}/(2\lambda)-\bigl(q\log(1-e^{-2\lambda c}) \\
&\qquad + (1-q)\log(1-e^{-2\lambda(1+c)})\bigr)/(2\lambda)
\end{aligned}$
\\

Vincze-Le Cam
& $\displaystyle (2-2t)/(t+1)$
& $\displaystyle 2(p-q)^2/((p+q)\bigl(2-p-q\bigr))\leq \mathrm{VC}$
\\
&
&
$\displaystyle \begin{aligned}[t]
&\text{PWG: } p\leq 2(1+c)-q  -4\Bigl(q\sqrt{c}+(1-q)\sqrt{1+c}\Bigr)^2/\bigl(2+\mathrm{VC} \bigr) 
\end{aligned}$
\\
\hline
\end{tabular}

\label{tab:com_event_specialized_with_dpi_full}
\end{table}

\section{DPI and Optimality}
\label{app::DPI_optimal}

We first prove a general theorem. 
\begin{theorem}
Let $(\mathcal X,\mathcal F)$ and $(\mathcal Y,\mathcal G)$ be measurable spaces, let
$P,Q$ be probability measures on $(\mathcal X,\mathcal F)$, and let
$\phi:\mathcal X\to\mathcal Y$ be measurable. Define the pushforwards
\[
P_\phi := P\circ \phi^{-1}, \qquad Q_\phi := Q\circ \phi^{-1}.
\]
Define the restricted class \[
\mathcal T_\phi := \{\,T:\mathcal X\to\mathbb R:\ T=g\circ \phi \text{ for some measurable }g:\mathcal Y\to\mathbb R\,\}.
\]
For $f:(0,\infty)\to \mathbb R\cup\{+\infty\}$ being convex with $f(1)=0$ with its convex conjugate $f^*(t)$, we have  \[
\sup_{T\in\mathcal T_\phi} \Bigl\{\mathbb E_P[T]-\mathbb E_Q[f^*(T)]\Bigr\} = D_f(P_\phi\Vert Q_\phi).
\]
\end{theorem}

\begin{proof}
Take any $T\in\mathcal T_\phi$. Then $T=g\circ \phi$ for some measurable
$g:\mathcal Y\to\mathbb R$. Therefore
\[ 
\mathbb E_P[T] = \mathbb E_P[g(\phi(X))] = \mathbb E_{P_\phi}[g].
\]
Also, \[\mathbb E_Q[f^*(T)]=\mathbb E_Q[f^*(g(\phi(X)))]=\mathbb E_{Q_\phi}[f^*(g)].
\]
Hence, \[ 
\mathbb E_P[T]-\mathbb E_Q[f^*(T)]=\mathbb E_{P_\phi}[g]-\mathbb E_{Q_\phi}[f^*(g)].
\]
Taking the supremum over all $T\in\mathcal T_\phi$ is therefore equivalent to
taking the supremum over all measurable $g$ on $\mathcal Y$. Thus, \[
\sup_{T\in\mathcal T_\phi} \Bigl\{\mathbb E_P[T]-\mathbb E_Q[f^*(T)]\Bigr\}=\sup_g \Bigl\{ \mathbb E_{P_\phi}[g]-\mathbb E_{Q_\phi}[f^*(g)]\Bigr\}=D_f(P_\phi\Vert Q_\phi).
\]
Since $\mathcal T_\phi$ is a subclass of all measurable test functions on
$\mathcal X$, we also have
\[
D_f(P\Vert Q)=\sup_T\Bigl\{\mathbb E_P[T]-\mathbb E_Q[f^*(T)]
\Bigr\} \ge \sup_{T\in\mathcal T_\phi}\Bigl\{\mathbb E_P[T]-\mathbb E_Q[f^*(T)]\Bigr\}.
\]
Combining the two displays yields
\[
D_f(P\Vert Q)\ge D_f(P_\phi\Vert Q_\phi).
\]
\end{proof}

This theorem induces the following corollary. 
\begin{corollary}
Let $\{A_1,\dots,A_n\}$ be a measurable partition of $\mathcal X$. Restrict the
variational class to functions of the form
\[
T=\sum_{i=1}^n a_i \mathds 1_{A_i}.
\]
Then
\[
\sup_{a_1,\dots,a_n}\left\{\sum_{i=1}^n a_i P(A_i)-\sum_{i=1}^n Q(A_i)f^*(a_i)\right\} = D_f\bigl((P(A_i))_{i=1}^n\Vert (Q(A_i))_{i=1}^n\bigr).
\]
Consequently, \[
D_f(P\Vert Q)\ge D_f\bigl((P(A_i))_{i=1}^n\Vert (Q(A_i))_{i=1}^n\bigr).
\]
\end{corollary}

\section{Details on Data Memorization}
\label{app::details_mem}

\subsection{Discussions and Comparison on \cite[Theorem~5]{sefidgaran2025tighter}}

We first discuss how do we improve the Fano step in the proof of \cite[Theorem~5]{sefidgaran2025tighter}. 
 
Let
\[
b_t:=\lfloor nt\rfloor,
\qquad
q_t:=2^{-n}\sum_{k=0}^{b_t}\binom{n}{k},
\qquad
L_t:=\log\frac1{q_t}.
\]
Write $I_n:=I(W;J| \tilde{Z})$ and assume it is positive, take $c_t^\star:=W_0\big(\frac{I_n}{q_t}\big)$ where $W_0$ is the principal branch of the Lambert $W$ function.

Since
$q_t e^{c_t^\star}=I_n/c_t^\star$, \eqref{eq::selector_KL_bound} gives
\begin{equation}
p_t
:=
\mathbf P\big(d_H(\hat J,J)\le nt\big)
\le
\Psi_n(t)
:=
\frac{
I_n+\log\left(1-q_t+\frac{I_n}{c_t^\star}\right)
}{
c_t^\star
}.
\label{eq::selector_Psi_bound}
\end{equation}

For fixed $t\in[0,1/2)$,
\[
L_t=n(\log 2-h_2(t))+O(\log n).
\]
In particular, when $I_n=O(1)$ and $I_n>0$,
\[
c_t^\star
=
L_t-\log\frac{L_t}{I_n}+O(1),
\]
and hence
\begin{equation}
\Psi_n(t)
=
\frac{I_n}{L_t}
+
O\left(\frac{\log n}{n^2}\right).
\label{eq::selector_Psi_asymptotic}
\end{equation}

By contrast, the Fano step  in
\cite[Theorem~5]{sefidgaran2025tighter} gives
\[
\frac{I_n+\log 2}{L_t}
=
\frac{I_n}{L_t}
+
\frac{\log 2}{L_t}.
\]

Thus, in the low-CMI regime $I_n=O(1)$, the sharpened KL event bound removes
the leading $\Theta(1/n)$ slack coming from the $\log 2$ term. 

We then compare with the statement of \cite[
Theorem~5(ii)]{sefidgaran2025tighter}, which gave the following false-positive conclusion. 
Under the same low-CMI condition, if
\[
\mathbf P(T_n\ge \alpha n)\ge q,
\]
then for every $\epsilon\in(0,\alpha)$,
\[
\mathbf P(F_n\ge m_\epsilon)
\ge
(\alpha-\epsilon)q,
\qquad
m_\epsilon
=
\frac{\epsilon}{1/q+\epsilon-\alpha}n-o(n).
\]
Our sharpened selector-recovery bound strengthens this conclusion. Indeed,
for every fixed $\beta<\alpha$ and admissible $\eta>0$, it gives
\[
\mathbf P(F_n>\beta n)
\ge
q
-
\frac{
\Psi_n(t_{\alpha,\beta,\eta})
}{
1-\exp\left(-\frac{2\eta^2}{1-\alpha+\beta}n\right)
},
\qquad
t_{\alpha,\beta,\eta}
=
\frac{1-\alpha+\beta}{2}+\eta .
\]
Consequently, if $I(W;J| \tilde{Z})=o(n)$, then
$\Psi_n(t_{\alpha,\beta,\eta})\to0$, and hence
\[
\mathbf P(F_n>\beta n)\ge q-o(1),
\qquad
\forall\,\beta<\alpha .
\]

Compared with \cite[
Theorem~5(ii)]{sefidgaran2025tighter}, the improvement
is twofold. First, the probability lower bound improves from
$(\alpha-\epsilon)q$ to $q-o(1)$. Second, the false-positive threshold can be
taken to be any $\beta n$ with $\beta<\alpha$, whereas their threshold is
\[
m_\epsilon
=
\frac{\epsilon}{1/q+\epsilon-\alpha}n-o(n),
\]
whose leading coefficient is strictly smaller than $\alpha$. Therefore, for
any fixed $\epsilon\in(0,\alpha)$, we may choose $\beta$ such that
\[
\frac{\epsilon}{1/q+\epsilon-\alpha}<\beta<\alpha.
\]
Then, for all sufficiently large $n$,
\[
\{F_n>\beta n\}\subseteq \{F_n\ge m_\epsilon\},
\]
and our result implies
\[
\mathbf P(F_n\ge m_\epsilon)\ge q-o(1),
\]
which is stronger than the original lower bound
\[
\mathbf P(F_n\ge m_\epsilon)\ge(\alpha-\epsilon)q.
\]

\subsection{Proof of Corollary~\ref{cor::selector_renyi_threshold}}
\label{subapp::proof_selector_renyi_threshold}

\begin{proof}
Using DPI with the event
\[
E_{b_n}:=\{( \hat{J}  ,j): d_H(\hat{J} ,j)\le b_n\}
\]
and the $f$ corresponding to power-$\beta$ divergence, then also using \eqref{eq::PQ_Hellinger_divergence2}, we have
\[
p_{b_n}^\beta q_{b_n}^{1-\beta}
+
(1-p_{b_n})^\beta(1-q_{b_n})^{1-\beta}
\le
1+(\beta-1)\mathcal H_\beta(P_0\Vert Q_0).
\]
Dropping the nonnegative second term yields
\[
p_{b_n}
\le
q_{b_n}^{\frac{\beta-1}{\beta}}
\big(1+(\beta-1)\mathcal H_\beta(P_0\Vert Q_0)\big)^{1/\beta}.
\]
Using~\eqref{eq::Hellinger_to_Renyi},
\[
1+(\beta-1)\mathcal H_\beta(P_0\Vert Q_0)
=
\exp\big((\beta-1)D_\beta(P_0\Vert Q_0)\big),chec
\]
which proves~\eqref{eq::selector_renyi_threshold_explicit}. 
The remaining claims follow from $q_{b_n}\le e^{-nC_\tau}$.
\end{proof}

\subsection{Proof of Proposition~\ref{prop::selector_recovery_approx_max_info}}

\begin{proof}
Let
\[
P:=P_{JY},
\qquad
Q:=P_JP_Y.
\]
By Lemma~\ref{lem:approxmaxinfo_Egamma_equiv},
\[
I_\infty^\tau(J;Y)\leq \log \gamma
\quad\Longleftrightarrow\quad
E_\gamma(P\Vert Q)\leq \tau.
\]
Apply Proposition~\ref{prop::E_gamma_converse} to the event
\[
E_b:=\{(j,y):d_H(\phi(y),j)\leq b\}.
\]
Since $Q(E_b)=q_b$, we obtain
\[
p_b\leq \gamma q_b + E_\gamma(P\Vert Q)\leq \gamma q_b+\tau,
\]
which proves~\eqref{eq::selector_recovery_approx_max_info}.

\end{proof}

Consequently, if there exist sequences $(\gamma_n)$ and $(\tau_n)$ such that
\[
I_\infty^{\tau_n}(J;Y)\leq \log \gamma_n,
\qquad
\tau_n\to 0,
\qquad
\log \gamma_n=o(n),
\]
then for every fixed $\rho\in[0,1/2)$, we have $\mathbf{P} \big(\frac{1}{n}d_H(\hat{J} ,J)\leq \rho\big)\to 0$.


\end{document}